\documentclass[final,3p,times]{elsarticle}

\usepackage{amssymb}

\usepackage{graphicx}
\usepackage{dcolumn}
\usepackage{bm}
\usepackage{mathptmx}
\usepackage{array,color}
\usepackage{amsmath}

\newcommand{\ii}{\mathrm{i}}
\newcommand{\ee}{\mathrm{e}}

\newcommand{\cO}{\mathcal{O}}

\newtheorem{theorem}{Theorem}

\newtheorem{prop}{Proposition}
\newtheorem{proof}{Proof}
\newtheorem{remark}{Remark}

\begin{document}
\begin{frontmatter}
	
	
	
	\title{Vector rogue wave patterns of the multi-component nonlinear Schr{\"o}dinger equation and generalized mixed Adler--Moser polynomials}
	
	
	\author[aff1]{Huian Lin}
	
	\author[aff1]{Liming Ling$^{*}$}
	\ead{linglm@scut.edu.cn}
	\address[aff1]{School of Mathematics, South China University of Technology,Guangzhou, 510641, China}

\begin{abstract}
	This paper investigates the asymptotic behavior of high-order vector rogue wave (RW) solutions for any multi-component nonlinear Schr\"odinger equation (denoted as $n$-NLSE) with multiple internal large parameters and reports some new RW patterns, including non-multiple root (NMR)-type patterns with shapes such as $ 180 $-degree sector, jellyfish-like, and thumbtack-like shapes, as well as multiple root (MR)-type patterns characterized by right double-arrow and right arrow shapes. We establish that these RW patterns are intrinsically related to the root structures of a novel class of polynomials, termed generalized mixed Adler--Moser (GMAM) polynomials, which feature multiple arbitrary free parameters. The RW patterns can be understood as straightforward expansions and slight shifts of the root structures for the GMAM polynomials to some extent. In the \((x,t)\)-plane, they asymptotically converge to a first-order RW at the position corresponding to each simple root of the polynomials and to a lower-order RW at the position associated with each multiple root. Notably, the position of the lower-order RW within these patterns can be flexibly adjusted to any desired location in the \((x,t)\)-plane by tuning the free parameters of the corresponding GMAM polynomials.
	
\end{abstract}
\begin{keyword}
	$ n $-NLSE \sep RW pattern \sep Asymptotics \sep GMAM polynomial \sep Root structure 




\end{keyword}

\end{frontmatter}



\section{Introduction}

Rogue waves (RWs), also known as freak waves, monster waves, extreme waves, killer waves, and so on, were initially used to describe ocean waves with high amplitudes that appear suddenly from nowhere and disappear without a trace. Over the past few decades, the study of RWs has extended beyond oceanography, finding applications in diverse fields such as nonlinear optics \cite{2007Optical, Kibler2010, akhmediev2011, 2014Lecaplain}, plasma physics \cite{2011plasma, tsai2016, 2017plasma, merriche2017}, Bose--Einstein condensates \cite{pethick2008, carretero2008, 2009Matter, Manikandan2014}, and even finance \cite{yan2010, yan2011vector}. 
In physics, there are various mechanisms for the RW generation, such as Benjamin--Feir instability (i.e., modulation instability) \cite{shats2010, Pelinovsky2016}, three-dimensional (3D) directional wave focusing \cite{fochesato2007}, nonlinear focusing of transient frequency modulation wave groups \cite{pelinovsky2008large}, and dispersive focusing of unidirectional wave packets \cite{pelinovsky2011scenario}. {In mathematics, RWs can be modeled using a variety of integrable nonlinear partial differential equations. Prominent examples include the nonlinear Schr\"odinger equation (NLSE) \cite{peregrine1983water, yangb2021n, linh2024c}, the derivative NLSE \cite{xu2011, chan2014rogue}, Hirota equation \cite{ankiewicz2010r, ankiewicz2010d}, Ablowitz--Ladik equation \cite{ankiewicz2010d, wen2018m}, and so on. Meanwhile, multi-component integrable systems have also been employed to study RWs, such as the Manakov system \cite{baronio2012solutions, chen2015vector, ling2014high, yangb2023} (i.e., the couple NLSE), the coupled Fokas--Lenells equations \cite{ye2019general, suh2023}, multi-component NLSE \cite{zhangg2021, zhang2021parity, wanglh2022, zhangg2022}, multi-component derivative NLSE \cite{linh2024a}, and so on.}

Furthermore, the study of RW patterns has become a focal point in RW theory. Understanding these patterns allows for the prediction of future RW structures using initial waveforms. {For instance, in $2021$, Yang et al. established a connection between the high-order RW solutions and the root structures of the Yablonskii--Vorob'ev polynomial hierarchy \cite{yangb2021n, yangb2021}}. Then, Yang and our group all discovered RW patterns in two-component integrable systems, including the Manakov system \cite{yangb2023} and the coupled Fokas--Lenells equations \cite{suh2023}, where these patterns correspond to the root structures of the Okamoto polynomials. In $2023$, Zhang et al. explored RW patterns for the three-component and four-component NLSE \cite{zhangg2022}, corresponding to the root structures of generalized Wronskian--Hermite polynomials with jump parameters of $ 4 $ and $ 5 $, respectively. In 2024, we also studied the RW patterns of any $ n $-component derivative NLSE \cite{linh2024a}, corresponding to the root structures of generalized Wronskian--Hermite polynomials with a jump parameter of $ n+1$. These RW patterns only contain a single internal large parameter. In another recent development, Yang et al. studied the RW patterns of the NLSE corresponding to the root structures of Adler--Moser polynomials \cite{yangb2024}, where these patterns contain multiple internal large parameters and all roots of the polynomials are simple. Subsequently, we explored the RW patterns of the NLSE with multiple internal large parameters, corresponding to the root structures of Adler--Moser polynomials with multiple roots \cite{linh2024c}. However, to the best of our knowledge, RW patterns involving multiple internal large parameters have not yet been explored in multi-component integrable systems. 

This paper aims to address this gap by investigating vector RW patterns for a general multi-component NLSE (alias $ n $-NLSE) with multiple internal large parameters. To achieve this, we introduce a novel class of specialized polynomials with multiple free parameters, termed generalized mixed Adler--Moser (GMAM) polynomials. Our findings reveal that the \(n\)-NLSE exhibits many new RW patterns, emerging as multiple internal parameters in the high-order vector RW solutions grow large. The shapes of these patterns can be asymptotically predicted by the root structures of the GMAM polynomials. Interestingly, these patterns often appear as simple dilations and slight shifts of the root structures of the GMAM polynomials. Compared to Adler--Moser polynomials, GMAM polynomials exhibit significantly richer root structures, which lead to a wider variety of RW patterns. As a result, we report some new RW patterns in the \(n\)-NLSE, such as non-multiple root (NMR)-type patterns with shapes such as $ 180 $-degree sector, jellyfish-like, and thumbtack-like shapes, as well as multiple root (MR)-type patterns with right double-arrow and right arrow shapes. Based on the root structures of the GMAM polynomials, we analyze the asymptotics of these RW patterns and validate them against the actual RW solutions, demonstrating excellent agreement.

The rest of this paper is organized as follows. In Sec. \ref{Sec2}, we provide the formula for high-order vector rogue wave (RW) solutions of the $ n $-NLSE. In Sec. \ref{Sec_gwhp}, we introduce the GMAM polynomials and their various NMR and MR structures. In Sec. \ref{Sec_rwp}, we analyze asymptotic behaviors and pattern structures of the vector high-order RW solutions for the $ n $-NLSE. In Sec. \ref{Sec_exam}, we present examples of RW patterns with multiple internal large parameters for the $ 3 $-NLSE and $ 4 $-NLSE. In Sec. \ref{Sec_proof}, we provide the proofs of the main results in this paper. Finally, conclusions and discussions are presented in Sec. \ref{Sec_Con}.

\section{Model and vector RW solution}\label{Sec2}
	
The $n$-NLSE \cite{Ablowitz2004} is 
\begin{equation}\label{nNLSE}
	\begin{aligned}
		&\ii \mathbf{q}_{t}+\frac{1}{2}\mathbf{q}_{xx}+ \| \mathbf{q}\|_{2}^{2}\mathbf{q} =\mathbf{0}, \quad \mathbf{q}=\left[ q_{1}(x,t),q_{2}(x,t),\cdots,q_{n}(x,t)\right] ^{T}, 
	\end{aligned}
\end{equation}
where $ \mathbf{q} $ is the complex vector field, the subscripts $ x $ and $ t $ denote the partial derivatives concerning the variables, and $ \|\cdot\|_{2} $ is the standard Euclidean norm. 

Referring to Refs. \cite{zhangg2021, zhang2021parity}, we can similarly employ the generalized Darboux transformation method to generate high-order vector RW solutions of the $n$-NLSE from the plane wave seed solution
\begin{equation}\label{seed}
	\begin{aligned}
		&\mathbf{q}^{[0]}(x,t)=\left[{q}_{1}^{[0]}(x,t), {q}_{2}^{[0]}(x,t), \dots, {q}_{n}^{[0]}(x,t)\right]^{T},\quad {q}_{k}^{[0]}(x,t)=a_{k}e^{\ii \theta_{k}},\\
		&\theta_{k}=b_{k}x+ \left( -\frac{{b_{k}}^{2}}{2} + \| \mathbf{a}\|_{2}^{2} \right) t, \quad 1\leq k\leq n,
	\end{aligned}	
\end{equation}
where $\mathbf{a}=(a_{1}, a_{2}, \dots, a_{n})$, $ a_{k}\neq0$, and the real parameters $ a_{k} $ and $ b_{k} $ are the amplitudes and wave numbers, respectively. In particular, when the constants $a_{j}$ and $b_{j}$ satisfy the following constraint conditions \cite{zhangg2021}:
\begin{equation} \label{abk} 
	\begin{aligned}
		a_{k}=\frac{2\Im(\lambda_{0})}{n+1}\csc(\frac{k\pi}{n+1}),\quad
		b_{k}=\frac{2\Im(\lambda_{0})}{n+1}\cot(\frac{k\pi}{n+1})-2\Re(\lambda_{0}), 
	    \quad 1\leq k\leq n,
	\end{aligned}
\end{equation}
with $ \Re(\lambda_{0}) $ and $ \Im(\lambda_{0}) $ being the real and imaginary parts of the eigenvalue $ \lambda_{0} $, respectively, and $ \Im(\lambda_{0})\ne 0 $, we can obtain an $ (n+1) $-multiple root
\begin{equation}\label{chi0}
	\chi_{0}=2\Re(\lambda_{0})+ \ii\frac{2\Im(\lambda_{0})}{n+1},
\end{equation} 
of the corresponding characteristic equation
\begin{equation} \label{cheq} 
	\begin{aligned}
		\left( \chi-2{\lambda}-\sum_{i=1}^{n}\frac{a_{i}^{2}}{\chi+b_{i}}\right) \prod_{j=1}^{n}(\chi+b_{j})=0,
	\end{aligned}
\end{equation}
at $\lambda=\lambda_{0}$. 

Next, we present the determinant representation of the vector RW solutions containing multiple free large parameters in the following theorem. For convenience, we first introduce the Schur polynomials $ S_{i}(\mathbf{z}) $ with $ \mathbf{z}=(z_{1}, z_{2}, \cdots) $ are given by the generating function
\begin{equation}\label{schur}
	\sum_{j=0}^{\infty} S_{j}(\mathbf{z})\varepsilon^{j}=\exp\left( \sum_{j=1}^{\infty} z_{j}\varepsilon^{j}\right),
\end{equation}
and other notations:
\begin{equation}\label{xpm1}
	\begin{aligned}
		&\mathbf{x}^{[l,\pm]}=(x_{1}^{[l,\pm]}, x_{2}^{[l,\pm]}, \cdots), \quad x^{[l,+]}_{j}=\alpha_{j}x+\beta_{j}\ii t+ d_{l,j}, \quad x^{[l,-]}_{j}=(x^{[l,+]}_{j})^{*}, \\ &\mathbf{h}_{1}=(h_{1,1},h_{1,2},\cdots), \quad
		{\mathbf{h}_{2}^{[k]}=(h_{2,1}^{[k]},h_{2,2}^{[k]},\cdots),}
	\end{aligned}
\end{equation}
where $ d_{l,j} $'s are the arbitrary complex constants, and $\alpha_{j} $, $\beta_{j} $ and $ h_{i,j} $ are defined by the following expansions
\begin{equation}\label{xpm2}
	\begin{aligned}
		&\ii\left( \chi(\varepsilon)-\lambda(\varepsilon)\right)  = \sum_{j=0}^{\infty}\alpha_{j}\varepsilon^{j}, \quad
		\left( \frac{\chi(\varepsilon)^{2}}{2} -\lambda(\varepsilon)^{2} -\|\mathbf{a}\|_{2}^{2}\right)  = \sum_{j=0}^{\infty}\beta_{j}\varepsilon^{j}, \\
		&\ln\left( \dfrac{(\chi_{0}-\chi_{0}^{*})}{\chi^{[1]}\varepsilon} \dfrac{(\chi(\varepsilon)-\chi_{0})}{(\chi(\varepsilon)-\chi_{0}^{*})}\right)  =\sum_{j=1}^{\infty}h_{1,j}\varepsilon^{j},  
		\quad \ln\left( \frac{\chi(\varepsilon)+b_{k}}{\chi_{0}+b_{k}}\right) =\sum_{j=1}^{\infty}{h_{2,j}^{[k]}}\varepsilon^{j},
	\end{aligned}
\end{equation}
and  
\begin{equation}\label{lamchi}
	\lambda(\varepsilon) =\lambda_{0} +\lambda^{[1]}\varepsilon^{n+1}, \quad \chi(\varepsilon)=\chi_{0}+ \sum_{i=1}^{\infty}\chi^{[i]}\varepsilon^{i}, 
\end{equation}
satisfy the characteristic equation \eqref{cheq} with the small perturbation parameter $ \varepsilon $ and the arbitrary nonzero complex constant $ \lambda^{[1]}$.

\begin{theorem}\label{Theo1}
		Given an $ n $-dimensional integer vector $ \mathcal{N}= [N_{1}, N_{2}, \cdots, N_{n}] $ with $N_{i}\geq 0$ $ (1 \leq i\leq n)$ and $ \sum_{i=1}^{n}N_{i}=N $, the $n$-NLSE \eqref{nNLSE} admits high-order vector RW solutions $\mathbf{q}^{\mathcal{N}}(x,t) = \left[{q}_{1}^{\mathcal{N}}(x,t),{q}_{2}^{\mathcal{N}}(x,t), \ldots, {q}_{n}^{\mathcal{N}}(x,t) \right]^{T} $ with
	\begin{equation}\label{horw}
		q_{k}^{\mathcal{N}}(x,t)=  q^{[0]}_{k}(x,t)\frac{\det(\mathbf{M}^{[k]})}{\det(\mathbf{M}^{[0]})}, \qquad 1\leq k \leq n,
	\end{equation}
	where $ q_{k}^{[0]}(x,t) $ is the plane wave solution given in Eqs. \eqref{seed} and \eqref{abk}, $ \mathbf{M}^{[s]} (0\leq s\leq n)$ are defined by the following $ n\times n $ block matrices
	\begin{equation}\label{msmat1}
		\begin{aligned}
			&\mathbf{M}^{[s]}=\begin{pmatrix} 
				M^{[s;{1},{1}]} &  M^{[s;{1},{2}]} & \cdots & M^{[s;{1},{n}]} \\
				M^{[s;{2},{1}]} &  M^{[s;{2},{2}]} & \cdots & M^{[s;{2},n]}\\
				\vdots & \vdots  & \ddots  & \vdots \\
				M^{[s;n,1]} &  M^{[s;n,2]} &  \cdots & M^{[s;n,n]}
			\end{pmatrix}_{N\times N},  \\
			&M^{[s;l_{1},l_{2}]}= \left( \tau^{[s;l_{1},l_{2}]}_{(n+1)(i-1)+l_{1},(n+1)(j-1)+l_{2}} \right)_{1\leq i\leq N_{l_{1}}, 1\leq j\leq N_{l_{2}}}, \quad  1\leq l_{1},l_{2}\le n,
		\end{aligned}
	\end{equation}
the elements $ \tau^{[s;l_{1},l_{2}]}_{i,j} $ are given by
\begin{equation}\label{tauij1}
	\begin{aligned}
		&\tau_{i,j}^{[0;l_{1},l_{2}]}=\sum_{v=0}^{\min{(i,j)}}C_{1}^{v} \, S_{i-v}(\mathbf{x}^{[l_{1},-]}+v \mathbf{h}_{1}^{*})\, S_{j-v}(\mathbf{x}^{[l_{2},+]}+v \mathbf{h}_{1}), \\
		&\tau_{i,j}^{[k;l_{1},l_{2}]}=\sum_{v=0}^{\min{(i,j)}}C_{1}^{v} \, S_{i-v}(\mathbf{x}^{[l_{1},-]}+v \mathbf{h}_{1}^{*}+{(\mathbf{h}_{2}^{[k]})}^{*} )\, S_{j-v}(\mathbf{x}^{[l_{2},+]}+v \mathbf{h}_{1}-{\mathbf{h}_{2}^{[k]}}), \quad 1\leq k\leq n,
	\end{aligned}
\end{equation}
the constant $ C_{1}= \frac{|\chi^{[1]}|^{2}}{|\chi_{0}-\chi_{0}^{*}|^{2}} $, and $ \chi_{0} $, $ S_{j} $, $\lambda^{[1]}, \chi^{[1]}$, $ \mathbf{x}^{[l,\pm]} $, $ \mathbf{h}_{1} $, and $ {\mathbf{h}_{2}^{[k]}} $ are defined by Eqs. \eqref{chi0}-\eqref{lamchi}.
\end{theorem}
The proof of this Theorem can be found in Appendix \ref{Appe_proofRWs}.

Note that the free internal large parameters $ d_{l,j} $ and the integer vector $ \mathcal{N} $ in the high-order RW solution \eqref{horw} are two main factors affecting the structure of RW patterns. From the proof of the asymptotics of RW patterns in Theorem \ref{theo_mrp} below, it can be found that the parameters $ d_{l, (n+1)j} $ $ (j\geq 1) $ do not affect the dynamical behavior of the RW patterns. Therefore, we can assume $ d_{l,(n+1)j}=0 $ $ (j\geq 1) $ without loss of generality. Moreover, if $ N_{i}=0 $ in $ \mathcal{N} $, the block matrix \eqref{msmat1} of RW solution \eqref{horw} does not possess the submatrices $ M^{[s;i,j]} $ and $ M^{[s;j,i]} $ with $ 1\leq j\leq n $. When taking $ \mathcal{N}_{1}=[1,0,0,\ldots,0] $, we apply the formula \eqref{horw} to obtain a first-order fundamental vector RW solution $ \mathbf{q}^{[1]}(x,t)= \left[ {q}_{1}^{[1]}(x,t), {q}_{2}^{[1]}(x,t), \cdots, {q}_{n}^{[1]}(x,t) \right] ^{T} $ of $ n $-NLSE \eqref{nNLSE} with the internal free complex parameter $ d_{1,1} $, where $ q_{k}^{[1]}(x,t)=q_{k}^{[0]}(x,t)\hat{q}_{k}^{[1]}\left( x-x_{0},t-t_{0} \right)  $ with $ q_{k}^{[0]}(x,t) $ given in Eq. \eqref{seed} and
\begin{equation}\label{q1}
\begin{aligned}
	&	\hat{q}_{k}^{[1]}(x,t)= 1+\dfrac{2\ii}{|\chi_{0}+b_{k}|^{2}}\dfrac{ (\Re(\chi_{0})+b_{k}) \left( x+\Re(\chi_{0})t  \right) - \Im(\chi_{0})^{2} t +\frac{\ii}{2} }{\left( x+\Re(\chi_{0})t \right)^{2} +\Im(\chi_{0})^{2}t^{2}+\frac{1}{4\Im(\chi_{0})^{2}}}, \\
	& (x_{0}, t_{0})=\left( -\frac{\Re(\chi_{0})}{\Im(\chi_{0})}\Re( \frac{d_{1,1}}{\chi^{[1]}}) -\Im( \frac{d_{1,1}}{\chi^{[1]}}) , \frac{1}{\Im(\chi_{0})}\Re( \frac{d_{1,1}}{\chi^{[1]}})\right) .
\end{aligned}
\end{equation}
When $ x\rightarrow\infty $ and $ t\rightarrow\infty $, the first-order RW $  |q_{k}^{[1]}(x,t)| \rightarrow |a_{k}| $. {Moreover, at its central position $ (x,t)=(x_{0}, t_{0}) $, the amplitude of the first-order RW $ |q_{k}^{[1]}(x,t)|  $  is}
\begin{equation}\label{maxamp1}
	\left|a_{k}\left(1-\frac{4\Im(\chi_{0})^2}{ |\chi_{0}+b_{k}|^{2}}  \right)\right|.
\end{equation}

\section{Generalized mixed Adler--Moser polynomial}\label{Sec_gwhp}

As reported in Refs. \cite{yangb2021, yangb2023, zhangg2022, suh2023, yangb2024, linh2024a}, the patterns of high-order RW solutions for many $ (1+1) $-dimensional integrable equations (counting multi-component systems) are closely related to root structures of the particular Wronskian polynomials. Thus, to effectively study patterns and asymptotic behavior of the high-order RW solutions for the $ n $-NLSE, we will introduce a new class of special polynomials with multiple free parameters, referred to as the generalized mixed Adler--Moser (GMAM) polynomials.
	
Let $ \varphi_{l,j}(z) $ be the special Schur polynomials defined by
\begin{equation}\label{schpj}
	\sum_{j=0}^{\infty} \varphi_{l,j}(z) \varepsilon^{j}=\exp\left(z\varepsilon+ \sum_{\substack{i=1 \\ (n+1) \nmid i }}^{\infty}\kappa_{l,i}\varepsilon^{i} \right), \quad l=1,2,\ldots, n,
\end{equation}
where $ \varphi_{l,j}(z) = 0 $ for $ j<0 $, $ \varphi^{\prime}_{l,j+1}(z)=\varphi_{l,j}(z)$, and $ \kappa_{l,i} $  are arbitrary complex constants. Then, the monic GMAM polynomials are defined as
\begin{equation}\label{gswhp}
	W_{\mathcal{N}}(z)=c_{\mathcal{N}}\, \mathcal{W}\left(\{\varphi_{1,(n+1)(j-1)+1}(z)\}_{j=1}^{N_{1}}, \{\varphi_{2,(n+1)(j-1)+2}(z)\}_{j=1}^{N_{2}},\ldots, \{\varphi_{n,(n+1)(j-1)+n}(z)\}_{j=1}^{N_{n}} \right),
\end{equation}
where $ c_{\mathcal{N}} $ is a real constant, $ \mathcal{W}\left( \{\varphi_{l,(n+1)(j-1)+l}(z)\}_{j=1}^{N_{l}} \right) =\mathcal{W}\left( \varphi_{l,l}(z), \varphi_{l,n+1+l}(z), \ldots, \varphi_{l,(n+1)(N_{l}-1)+l}(z) \right) $ $ (1\leq l \leq n) $ are the usual Wronskian determinants, and the index vector $ \mathcal{N} $ is defined by Theorem \ref{Theo1}. Since $  \varphi_{l,j}(z) $ is a $ j $th-order polynomial in the variable $z$, we can calculate the degree 
\begin{equation}\label{degpol}
	\Gamma= \sum_{l=1}^{n}\frac{N_{l}}{2}(n(N_{l}-1)+2l)-\sum_{1\leq l< k\leq n}N_{l}N_{k},
\end{equation}
of the GMAM polynomial $ W_{\mathcal{N}}(z) $ and the number $ \sum_{l=1}^{n}\left(n(N_{l}-1)+k_{l} \right) $ of free parameters with $ k_{l}=l $ for $ N_{l}\ne 0 $ and $ k_{l}=n $ for $ N_{l}= 0 $. Therefore, when all free parameters $ \kappa_{l,i} $ all equal to zero, we have $ W_{\mathcal{N}}(z)=z^{\Gamma} $.

Especially, when $ n=1 $, the GMAM polynomials reduce to the Adler--Moser polynomials \cite{am2009,yangb2024}. When $n=1$ (i.e., $n$-NLSE is the scalar NLSE) and all parameters $\kappa_{1,i}$ are zero except for one $\kappa_{1,2j+1}=-\frac{2^{2j}}{2j+1}$ $(j\geq 1)$, the GMAM polynomials reduce to the Yablonskii--Vorob'ev polynomial hierarchy \cite{ kajiwara1996, clarkson2003a, yangb2021}. When $n=2$, all parameters $\kappa_{l,i}$ are zero except for one $\kappa_{l,j}=1$ $ (l=1 $ or $ 2) $, and $\mathcal{N}=[N,0]$ or $[0, N]$, the GMAM polynomials reduce to the Okamoto polynomial hierarchy \cite{okamoto1986, clarkson2003b, yangb2023}. When $ n $ is even and $ \kappa_{l,i}=0 $ $ (i\ne 2) $, the GMAM polynomials can reduce to the symmetric Okamoto polynomials for $ \kappa_{l,2}=\frac{1}{2} $ $ (1\leq l\leq n) $ and the symmetric Hermite polynomials for $ \kappa_{l,2}=\frac{1}{2(2n+1)} $ $ (1\leq l\leq n) $, respectively \cite{Clarkson2008}. When all parameters $\kappa_{l,i}$ are zero except for one $\kappa_{l,j}=1$ $ (1\leq l\leq n) $ and only one nonzero integer $N_{l}$ in the index vector $\mathcal{N}$, the GMAM polynomials reduce to the generalized Wronskian--Hermite polynomials \cite{zhangg2022, linh2024a}. Therefore, we can regard the GMAM polynomial \eqref{gswhp} as a mixed generalization of Adler--Moser polynomials. The first few GMAM polynomials are presented in Appendix \ref{Appe_tGMAMp}.

Next, we will discuss the root structure of the GMAM polynomials with nonzero free parameters, which is crucial for our analysis of the asymptotics of RW patterns in the later text. Based on the arbitrariness of the free complex parameters $ \kappa_{l, i} $ in $ W_{\mathcal{N}}(z) $, the root structures of such polynomials are highly complex and diverse. In Refs. \cite{clarkson2003a, yangb2021}, Clarkson et al. investigated root structures of the Yablonskii--Vorob'ev polynomial hierarchy. In Refs. \cite{yangb2024, linh2024c}, Yang and our group have separately discussed the root structures of Adler--Moser polynomials under conditions where all roots are simple and multiple roots exist. In Ref. \cite{Clarkson2008}, Clarkson et al. presented the root structures of some symmetric Okamoto polynomials and symmetric Hermite polynomials, revealing stacked symmetric rhombuses or symmetric rectangles. Additionally, in Ref. \cite{linh2024a}, we have studied the root structures of the GMAM polynomials with only one nonzero integer $N_{l}$ in the index vector $\mathcal{N}$ and all parameters being zero except for one $\kappa_{l,j}=1$ $((n+1) \nmid j)$.

We categorize the infinitely many different root structures of the GMAM polynomials into two main types: non-multiple root (NMR) structure and multiple root (MR) structure. Here, we briefly present some examples to demonstrate the different NMR and MR structures of the polynomials $W_{\mathcal{N}}(z)$. For instance, when the index vector $ \mathcal{N}=[4,3,2] $ and the free parameter $\kappa_{l,j}$ equals to $ -l$ or $ (4-l) \, \ii^{j}$ with $ 1\leq l \leq 3 $, $ 1\leq j\leq 4(N_{l}-1)+l $, and $ 4 \nmid j $, we obtain two NMR structures of $W_{[4,3,2]}(z)$. As shown in Fig. \ref{Fig1} (i) and (ii), the distributions of the $ 20 $ simple roots of $W_{[4,3,2]}(z)$ resemble a $ 180 $-degree sector and a jellyfish-like shape, respectively.
When the index vector $ \mathcal{N}=[3,2,2,1] $ and the free parameter $\kappa_{l,j}$ equals to $ l$ or $ \frac{l}{2}\,\ii^{j}$ with $ 1\leq l \leq 4 $, $ 1\leq j\leq 5(N_{l}-1)+l $, and $ 5 \nmid j $, we generate two NMR structures of $W_{[3,2,2,1]}(z)$. As shown in Fig. \ref{Fig2} (i) and (ii), the distributions of the $ 14 $ simple roots of $W_{[3,2,2,1]}(z)$ resemble two distinct thumbtack-like shapes, respectively.

Furthermore, for the MR structures of GMAM polynomials with nonzero free parameters, we primarily discuss a special case characterized by the presence of only one multiple root $ z=z_{0}-\kappa_{1,1} $, with all other roots being simple. Here, the free parameters $ \kappa_{l,i} $ satisfy
\begin{equation}\label{mmrp}
\begin{aligned}
  	&\kappa_{l_{1},1}=\kappa_{l_{2},1}, \quad  l_{1}\ne l_{2}, \quad 1\leq l_{1},l_{2}\leq n,\\
  	&\kappa_{l,i}=\frac{z_{0}^{i}}{i}, \quad z_{0}\in \mathbb{C} \setminus \{0\}, \quad 1\leq l\leq n, \quad i> 1, \quad (n+1) \nmid i.
\end{aligned}
\end{equation}
However, for the arbitrary integer vector index parameters $\mathcal{N}$, the multiplicity of a root $z=z_{0}-\kappa_{1,1}$ of the polynomial $ W_{\mathcal{N}}(z) $ with free parameters \eqref{mmrp} exhibits considerable complexity and variability. Consequently, it is not feasible to directly provide a general formula for the multiplicity of the multiple root. Here, we present a formula for the multiplicity of the multiple root for $ W_{\mathcal{N}}(z) $ with free parameters \eqref{mmrp} and a specific type of index vector parameters $\mathcal{N}$ in the following theorem.

\begin{figure*}[!htbp]
	\centering
	\includegraphics[width=0.8\textwidth]{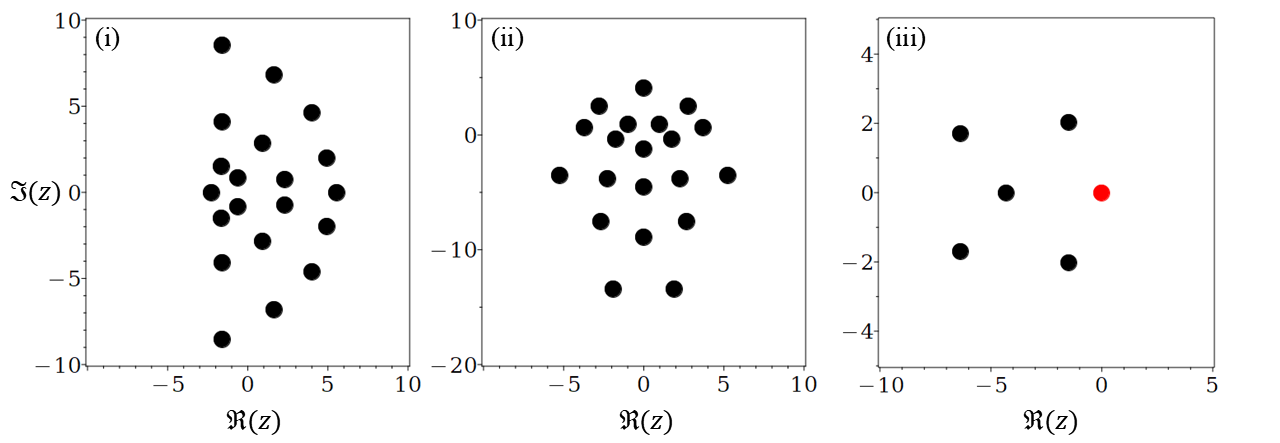}
	\caption{The root structures of the GMAM polynomials $ W_{[4,3,2]}(z) $ with the free parameter $ \kappa_{l,j}=-l $ in (i), $ \kappa_{l,j}=(4-l) \, \ii^{j}$ in (ii), and $ \kappa_{l,j}=\frac{1}{j}$ in (iii). Black points represent simple roots, while red point represents the $15$-multiple zero root.}
	\label{Fig1}
\end{figure*}

\begin{figure*}[!htbp]
	\centering
	\includegraphics[width=0.8\textwidth]{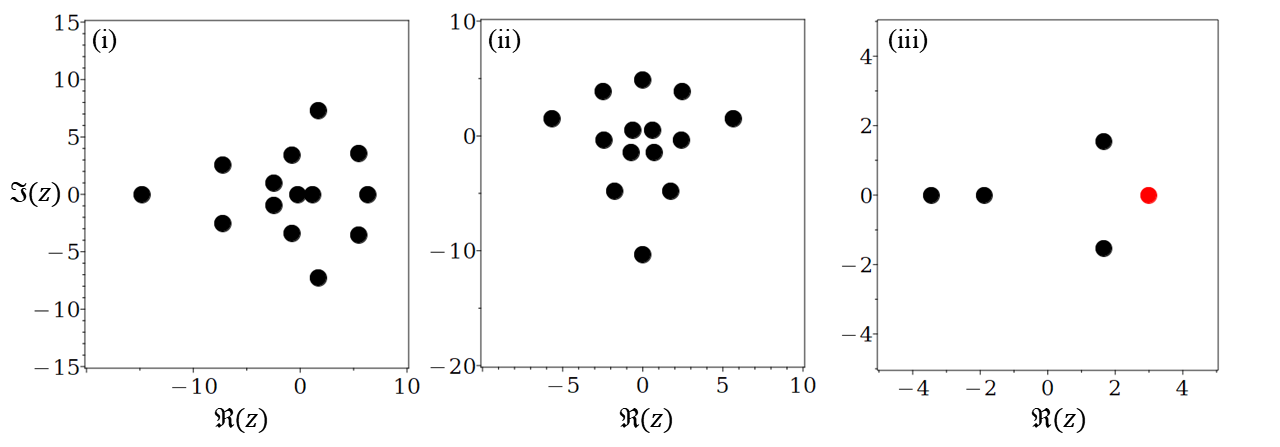}
	\caption{The root structures of the GMAM polynomials $ W_{[3,2,2,1]}(z) $ with the free parameter $ \kappa_{l,j}=l $ in (i), $ \kappa_{l,j}= \frac{l}{2}\,\ii^{j} $ in (ii), and $\kappa_{l,1}= -2 $ and $ \kappa_{l,j}=\frac{1}{j} (j>1)$ in (iii). Black points represent simple roots, while red point represents the $10$-multiple root.}
	\label{Fig2}s
\end{figure*}

\begin{theorem}\label{theo_mrs}
Assuming that the integer index vector $ \mathcal{N}= [N_{1}, N_{2}, \cdots, N_{n}] $ satisfies $ N_{1}\geq N_{2}\geq \ldots \geq N_{n} \geq 0  $ and $ \sum_{j=1}^{n} N_{j}=N$, and that the free parameters $ \kappa_{l,i} $ are defined in Eq. \eqref{mmrp}, the GMAM polynomial $ W_{\mathcal{N}}(z) $ has a $ \Gamma_{0} $-multiple root $ z=z_{0}-\kappa_{1,1} $. Here, $ \Gamma_{0} $ is defined by Eq. \eqref{degpol}, where $ N_{i} $ is replaced by $ \hat{N}_{i} $, and
\begin{equation}\label{gamma0}
	\left[ \hat{N}_{1}, \hat{N}_{2}, \cdots, \hat{N}_{n}\right] =\begin{cases}
		[{N}_{1}, \ldots, N_{i-1},N_{i}-1,N_{i+1},\ldots {N}_{n}], & N_{1}=N_{i}>N_{i+1}, \quad 1\leq i\leq n-1,\\
		[N_{1}, N_{2},\ldots, N_{n-1}, N_{n}-1], &  N_{1}=N_{n}.
	\end{cases}
\end{equation}

\end{theorem}
This Theorem will be proved in Sec. \ref{subsec_proofmrs}. 

\begin{remark}
	{In Theorem \ref{theo_mrs}, we only provide the formula for the multiplicity of the multiple root  $ z=z_{0}-\kappa_{1,1} $ of the GMAM polynomial $ W_{\mathcal{N}}(z) $ with free parameters \eqref{mmrp} under a specific integer index vector $ \mathcal{N} $, where the elements of  $ \mathcal{N} $ satisfy $ N_{1}\geq N_{2}\geq \ldots \geq N_{n} \geq 0  $. For other cases with the general integer index vector $ \mathcal{N} $, the multiplicity of the multiple root for $ W_{\mathcal{N}}(z)$ can also be computed by using a similar approach.} 
\end{remark}

On the other hand, we note that the root structure of the GMAM polynomials with free parameters \eqref{mmrp} may include multiple roots other than the multiple root $ z=z_{0}-\kappa_{1,1} $. However, in this paper, we focus on the case where all roots are simple except for the only one multiple root $ z=z_{0}-\kappa_{1,1} $. As shown in Figs. \ref{Fig1} (iii) and \ref{Fig2} (iii), this type of MR structures for $W_{[4,3,2]}(z)$ (with $ z_{0}=\kappa_{1,1}=1 $) and $W_{[3,2,2,1]}(z)$ (with $ z_{0}=1 $ and $ \kappa_{1,1}=-2 $) are displayed, respectively. Fig. \ref{Fig1} (iii) is a right double-arrow shape formed by five nonzero simple roots and a $ 15 $-multiple zero root of $W_{[4,3,2]}(z)$, while Fig. \ref{Fig2} (iii) depicts a right arrow shape formed by four nonzero simple roots and a $ 10 $-multiple nonzero root of $ W_{[3,2,2,1]}(z) $.

\section{{Asymptotic} analysis of RW patterns with multiple large internal parameters}\label{Sec_rwp}

In Sec. \ref{Sec2}, we assumed that the parameters $  d_{l,(n+1)j}=0 $ $ (1\leq l\leq n) $ for the vector RW solution \eqref{horw} of $n$-NLSE \eqref{nNLSE}, which implies that $\mathbf{q}^{\mathcal{N}}(x,t)$ contains the arbitrary internal complex parameters $\{d_{l,j}\}_{1\leq l \leq n, j\in \Lambda_{l}}$ with the sets
\begin{equation}\label{setlam}
	\Lambda_{l}=\{i| 1\leq i \leq (n+1)(N_{l}-1)+l, (n+1) \nmid i\}.
\end{equation}
When $n=1$ and the multiple internal parameters $d_{l,j}$ become large, the high-order RW solutions of the scalar NLSE exhibit patterns corresponding to the root structures of the Adler--Moser polynomials \cite{yangb2024, linh2024c}. In this section, we will discuss vector RW patterns for the vector RW solutions $\mathbf{q}^{\mathcal{N}}(x,t)$ \eqref{horw} of $n$-NLSE \eqref{nNLSE} and investigate their asymptotic behaviors, where these patterns, characterized by multiple large internal parameters, correspond to the root structures of the GMAM polynomials. We categorize the RW patterns of the $n$-NLSE into two main types: NMR-type and MR-type, corresponding to the NMR structure and MR structure of the GMAM polynomials, respectively.

\subsection{Asymptotics of NMR-type RW pattern}
 
Let the internal large parameters $\{d_{l,j}\}_{1\leq l \leq n, j\in \Lambda_{l}}$ in the $\mathcal{N}$-order vector RW solution $\mathbf{q}^{\mathcal{N}}(x,t) $ \eqref{horw} of the $n$-NLSE be the following form
\begin{equation}\label{dj1}
    d_{l,j}=\kappa_{l,j}A^{j},
\end{equation}
where $A\gg 1 $ is a large positive constant, $\{\kappa_{l,j}\}_{1\leq l \leq n, j\in \Lambda_{l}}$ are free complex constants (not all of which are zero). Then, when all roots of the GMAM polynomial $ W_{\mathcal{N}}(z) $ with parameters $\{\kappa_{l,j}\}_{1\leq l \leq n, j\in \Lambda_{l}}$ are simple, we obtain the NMR-type RW patterns of $n$-NLSE \eqref{nNLSE}. In the following theorem, we analyze the asymptotic behaviors of these RW patterns.

\begin{theorem}\label{theo_nmrp}
Let all roots of the GMAM polynomial $W_{\mathcal{N}}(z)$ with free parameters $\{\kappa_{l,j}\}_{1\leq l \leq n, j\in \Lambda_{l}}$ be simple, and let the internal large parameters $\{d_{l,j}\}_{1\leq l \leq n, j\in \Lambda_{l}}$ of the $\mathcal{N}$-order vector RW solution $\mathbf{q}^{\mathcal{N}}(x,t)$ be defined by Eq. \eqref{dj1}. Then, the NMR-type RW pattern of $q_{k}^{\mathcal{N}}(x,t)$ \eqref{horw} asymptotically splits into $\Gamma$ fundamental first-order RWs, expressed as $\hat{q}_{k}^{[1]}(x-\bar{x}_{0},t-\bar{t}_{0}) q_{k}^{[0]}(x,t)$, where $q_{k}^{[0]}(x,t)$, $\hat{q}_{k}^{[1]}(x,t)$, and $\Gamma$ are given in Eqs. \eqref{seed}, \eqref{q1}, and \eqref{degpol}, respectively. Here, $ (\bar{x}_{0}, \bar{t}_{0}) $ is defined by
\begin{equation}\label{xt01}
	\begin{aligned}
		&\ii \chi^{[1]}(\bar{x}_{0}+\chi_{0} \bar{t}_{0})=\bar{z}_{0}A -\Delta(\bar{z}_{0}),
	\end{aligned}
\end{equation}
$\Delta(\bar{z}_{0}) = \cO(1)$ is a constant defined by Eq. \eqref{Delz0}, $\bar{z}_{0}$ is a simple root of $W_{\mathcal{N}}(z)$, and $\chi_{0}$ and $\chi^{[1]}$ are given in Eqs. \eqref{chi0} and \eqref{lamchi}, respectively. 

For $\sqrt{(x-\bar{x}_{0})^{2}+(t-\bar{t}_{0})^{2}} = \cO(1)$ and $A\gg 1$, the RW solution $q_{k}^{\mathcal{N}}(x,t)$ admits the following asymptotics:
\begin{equation}\label{asym_nmrp}
    q_{k}^{\mathcal{N}}(x,t) = {\hat{q}_{k}^{[1]}}(x-\bar{x}_{0}, t-\bar{t}_{0}) q_{k}^{[0]}(x,t) +\cO(A^{-1}), \quad 1\leq k \leq n.
\end{equation}

Moreover, as $ A\rightarrow \infty $ and $ (x,t) $ is far from the positions of all $ (\bar{x}_{0}, \bar{t}_{0}) $, the vector RW solution $ \mathbf{q}^{\mathcal{N}}(x,t) $ \eqref{horw} asymptotically approaches the vector plane wave background $ \mathbf{q}^{[0]}(x,t) $ \eqref{seed}.
\end{theorem}

The proof of this Theorem is presented in Sec.  \ref{subsec_proofnmrp}.

Note that when the large parameter \( A\rightarrow \infty \), the influence of the constant term \( \Delta(\bar{z}_{0})\) on the position of the first-order RW in the RW patterns becomes negligible. Furthermore, based on the asymptotics of the vector RW solution $ \mathbf{q}^{\mathcal{N}}(x,t) $ in Theorem \ref{theo_nmrp}, we provide the following asymptotic expression for $ \left|q_{k}^{\mathcal{N}}(x,t)\right| $:
\begin{equation}\label{asym_nmrp2}
	\left|q_{k}^{\mathcal{N}}(x,t)\right| =a_{k}+ \sum_{(\bar{x}_{0}, \bar{t}_{0})} \left(\left|{q}^{[1]}_{k}(x-\bar{x}_{0}, t-\bar{t}_{0})\right|-a_{k} \right) +\cO(A^{-1}), \quad 1\leq k\leq n,
\end{equation}
where $ (\bar{x}_{0}, \bar{t}_{0}) $ are defined by Eq. \eqref{xt01} with $ \bar{z}_{0} $ traversing $ \Gamma $ simple roots of $W_{\mathcal{N}}(z)$, and $ a_{k} $ and $ \Gamma $ are given in Eqs. \eqref{abk} and \eqref{degpol}, respectively.

\subsection{Asymptotics of MR-type RW pattern}

This subsection will discuss the MR-type patterns of the high-order vector RW solution $ \mathbf{q}^{\mathcal{N}}(x,t) $ \eqref{horw} for $ n $-NLSE \eqref{nNLSE}, which correspond to MR structures of the GMAM polynomials $W_{\mathcal{N}}(z)$. The MR structures of $W_{\mathcal{N}}(z)$ are complex and diverse due to the influence of various free complex parameters $\{\kappa_{l,j}\}_{1\leq l \leq n, j\in \Lambda_{l}}$. Therefore, we primarily focus on a specific class of MR-type RW patterns, which corresponds to a particular MR structure for $W_{\mathcal{N}}(z)$, i.e., one with a \(\Gamma_{0}\)-multiple root and \(\Gamma-\Gamma_{0}\) simple roots. The values of $ \Gamma_{0} $ and $ \Gamma $ are given in Theorem \ref{theo_mrs} and Eq. \eqref{degpol}, respectively. Then, we can extend the asymptotic analysis of this specific MR-type RW pattern to other MR-type RW patterns.

Assume that $ z_{0}-\kappa_{1,1}=\ii \chi^{[1]}(x_{0}+2\ii t_{0})A^{-1} $ is only one $ \Gamma_{0} $-multiple root of $W_{\mathcal{N}}(z)$ with free parameters $ \{\kappa_{l,j}\}_{j\in \Lambda} $ and $ \Gamma_{0} $ given in Theorem \ref{theo_mrs}. For the high-order vector RW solution $ \mathbf{q}^{\mathcal{N}}(x,t) $ \eqref{horw}, when the internal large parameters $\{d_{l,j}\}_{1\leq l \leq n, j\in \Lambda_{l}}$ are defined by
\begin{equation}\label{dj2}
	d_{l,j}=
	\begin{cases}
		\kappa_{1,1}A, & j=1\\
		\kappa_{l,j}A^{j} +\hat{\kappa}_{j}, & j>1, \quad \hat{\kappa}_{j}=-(\alpha_{j}x_{0}+\ii \beta_{j}t_{0}), \quad j\in \Lambda, \quad A\gg 1,
	\end{cases} 
\end{equation}
with the parameters $ (\alpha_{j}, \beta_{j}) $ defined by Eq. \eqref{xpm2}, we construct the specific MR-type RW patterns. 
The structures of these patterns are divided into two different regions: the simple-root region and the multiple-root region, which correspond to the simple roots and multiple root of the relevant polynomial $W_{\mathcal{N}}(z)$, respectively. Thus, the NMR-type RW patterns can be regarded as having only a simple-root region, without a multiple-root region. Then, we present their asymptotics in the following theorem. 

\begin{theorem}\label{theo_mrp}
	Suppose that the integer index vector $ \mathcal{N}= [N_{1}, N_{2}, \cdots, N_{n}] $ satisfies $ N_{1}\geq N_{2}\geq \ldots \geq N_{n} \geq 0  $ and $ \sum_{i=1}^{n} N_{i}=N$, the internal large parameters $ \{d_{l,j}\}_{1\leq l \leq n, j\in \Lambda_{l}} $ in the $ \mathcal{N} $-order vector RW solution $ \mathbf{q}^{\mathcal{N}}(x,t) $ \eqref{horw} of $ n $-NLSE \eqref{nNLSE} are defined by Eq. \eqref{dj2}, and $ {z}_{0}-\kappa_{1,1}=\ii \chi^{[1]}({x}_{0}+\chi_{0} {t}_{0})A^{-1} $ is only one $ \Gamma_{0} $-multiple root of the polynomial $W_{\mathcal{N}}(z)$ with the free parameters \eqref{mmrp} and $ \Gamma_{0} $ defined in Theorem \ref{theo_mrs}. Then, the vector RW solution $ \mathbf{q}^{\mathcal{N}}(x,t) $ admits the following asymptotics:
	\begin{enumerate}
		
		\item[(1)] In the simple-root region with $ \sqrt{(x-{x}_{0})^{2}+(t-{t}_{0})^{2}} > \cO(A) $, the specific MR-type RW pattern of $ q_{k}^{\mathcal{N}}(x,t) $ $ (1\leq k\leq n) $ asymptotically splits into $ \Gamma-\Gamma_{0} $ first-order RWs, expressed as $\hat{q}_{k}^{[1]}(x-\bar{x}_{0},t-\bar{t}_{0}) q_{k}^{[0]}(x,t)$, where $ \Gamma $ is given in Eq. \eqref{degpol}, and follows the same asymptotic expressions as given in Eq. \eqref{asym_nmrp}. Here, $ (\bar{x}_{0}, \bar{t}_{0}) $ is defined by Eq. \eqref{xt01}, where $ \Delta(\bar{z}_{0}) $ is replaced by $ {\Delta} (\bar{z}_{0}) +\hat{\kappa}_{2}A^{-1} \frac{W_{\mathcal{N}}^{[2]}(\bar{z}_{0})}{W_{\mathcal{N}}'(\bar{z}_{0})} $,
		$ \bar{z}_{0} $ is a simple root of $W_{\mathcal{N}}(z)$, and $ W_{\mathcal{N}}^{[2]}(\bar{z}_{0}) $ is defined by Eq. \eqref{wnp2}.
		
		\item[(2)] In the multiple-root region with $ \sqrt{(x-{x}_{0})^{2}+(t-{t}_{0})^{2}} \leq \cO(A) $, the specific MR-type RW pattern of $ q_{k}^{\mathcal{N}}(x,t) $ $ (1\leq k\leq n) $ asymptotically approaches an $ \hat{\mathcal{N}} $-order RW, expressed as $ \hat{q}_{k}^{\hat{\mathcal{N}}}(x-{x}_{0}, t-{t}_{0}) q^{[0]}_{k}(x,t) $, where $ \hat{\mathcal{N}}=\left[ \hat{N}_{1}, \hat{N}_{2}, \cdots, \hat{N}_{n}\right] $ with $ \hat{N}_{l} $ $ (1\leq l\leq n) $ defined by Eq. \eqref{gamma0}, $ q^{[0]}_{k}(x,t) $ is given in Eq. \eqref{seed}, and $ \hat{q}_{k}^{\hat{\mathcal{N}}} (x-{x}_{0}, t-{t}_{0}) ={q}_{k}^{\hat{\mathcal{N}}} (x-{x}_{0},t-{t}_{0}) \left( q_{k}^{[0]}(x-{x}_{0},t-{t}_{0})\right)^{-1} $ with the internal parameters $ d_{l,j}=h_{1,j} $ $ (1\leq l\leq n, j\geq 1) $ given in Eq. \eqref{xpm2}. In addition, when $ \sqrt{(x-{x}_{0})^{2} +(t-{t}_{0})^{2}} = \cO(1)$ and $A\gg 1$, $ q_{k}^{\mathcal{N}}(x,t) $ exists the following asymptotic expression:
		\begin{equation}\label{asym_mrp1}
			q_{k}^{\mathcal{N}}(x,t) =\hat{q}_{k}^{\hat{\mathcal{N}}}(x-{x}_{0}, t-{t}_{0}) q^{[0]}_{k}(x,t) +\cO(A^{-1}).
		\end{equation}

		\item[(3)] When $ A\rightarrow \infty $ and $ (x,t) $ is far from the positions of all $ (x_{0}, t_{0}) $ and $ (\bar{x}_{0}, \bar{t}_{0}) $, the vector RW solution $ \mathbf{q}^{\mathcal{N}}(x,t) $ asymptotically approaches the vector plane wave background $ \mathbf{q}^{[0]}(x,t) $ \eqref{seed}.
		
	\end{enumerate}
\end{theorem}
The proof process is outlined in Sec. \ref{subsec_proofmrp}.

\begin{remark}
	When the parameter $ \kappa_{1,1}=z_{0} $ in Theorem \ref{theo_mrp}, the $ \Gamma_{0} $-multiple root of the GMAM polynomial $W_{\mathcal{N}}(z)$ with free parameters \eqref{mmrp} is $ z=0 $. This implies that the position $ (x_{0}, t_{0}) $ of the lower-order RW in the NMR-type RW pattern is at the origin of the $ (x,t) $-plane. Moreover, the parameters $ \hat{\kappa}_{j}=0 $ $ (j>1) $ in the internal large parameters \eqref{dj2} of the vector RW solution $ \mathbf{q}^{\mathcal{N}}(x,t) $, and the terms $ {\Delta}(\bar{z}_{0}) $ in the definitions of $ (\bar{x}_{0}, \bar{t}_{0}) $ in Theorems \ref{theo_nmrp} and \ref{theo_mrp} are identical. However, when the parameter $ \kappa_{1,1}\ne z_{0} $ in Theorem \ref{theo_mrp}, the position of the lower-order RW in the MR-type RW patterns is related to the nonzero multiple root $z=z_{0}-\kappa_{1,1}$ of the polynomial $W_{\mathcal{N}}(z)$. Then, we can shift the lower-order RW in the patterns to any position on the $(x, t)$-plane other than the origin by selecting appropriate nonzero complex parameters $z_{0}$ and $ \kappa_{1,1} $.
\end{remark}

Furthermore, according to the asymptotics of the vector RW solution $ \mathbf{q}^{\mathcal{N}}(x,t) $ described in Theorem \ref{theo_mrp}, we present an asymptotic expression for $ \left|q_{k}^{\mathcal{N}}(x,t)\right| $ below:
\begin{equation}\label{asym_nmrp3}
	\left|q_{k}^{\mathcal{N}}(x,t)\right| = \left|{q}_{k}^{\hat{\mathcal{N}}}(x-{x}_{0}, t-{t}_{0})\right|+ \sum_{(\bar{x}_{0}, \bar{t}_{0})} \left(\left|{q}^{[1]}_{k}(x-\bar{x}_{0}, t-\bar{t}_{0})\right|-a_{k} \right) +\cO(A^{-1}), \quad 1\leq k\leq n,
\end{equation}
where $ a_{k} $ $(1\leq k\leq n)$ are given in Eq. \eqref{abk}, $ (\bar{x}_{0}, \bar{t}_{0}) $ are determined by Eq. \eqref{xt01} with $ \bar{z}_{0} $ traversing $ \Gamma-\Gamma_{0} $ simple roots of $W_{\mathcal{N}}(z)$, and $ {q}_{k}^{\hat{\mathcal{N}}}(x-{x}_{0}, t-{t}_{0}) $, $ \Gamma $, and $ \Gamma_{0} $ are defined by Theorem \ref{theo_mrp}.

\section{Examples of RW patterns with multiple internal large parameters}\label{Sec_exam}

This section presents examples of RW patterns for the n-NLSE \eqref{nNLSE} with multiple large parameters, including the NMR-type and MR-type patterns defined in Sec \ref{Sec_rwp}. Here, we consider the RW patterns of the {$ 3 $}-NLSE and the {$4$}-NLSE under different internal large parameter conditions, respectively. For computational convenience, we set $ \lambda_{0}=\ii $ without loss of generality.  In addition, we will provide dynamical evolution plots of these RW patterns and the predicted positions of all wave peaks in the patterns. This also numerically verifies Theorems \ref{theo_nmrp} and \ref{theo_mrp}.

\begin{enumerate}
	\item[(1).] \textbf{Cases of the $ 3 $-NLSE}
	
	For the $ 3 $-NLSE, when the parameters $ a_{k} $ and $ b_{k} $ in the seed solution \eqref{seed} are
	\begin{equation}\label{abk2}
		\begin{aligned}
		(a_{1}, a_{2},a_{3})= \left( \frac{\sqrt{2}}{2}, \frac{1}{2}, \frac{\sqrt{2}}{2} \right) , \quad (b_{1}, b_{2}, b_{3})= \left( \frac{1}{2}, 0, -\frac{1}{2}\right),
		\end{aligned}
	\end{equation}
	the characteristic equation \eqref{cheq} exists a $ 4 $-multiple root $ \chi_{0}=\frac{\ii}{2} $ at $ \lambda=\ii $. We take the arbitrary parameters $ \lambda^{[1]}=2\ii $ and $ \chi^{[1]}=-\ii $. Then, when the internal large parameters $ \{d_{l,j}\}_{1\leq l \leq n, j\in \Lambda_{l}} $ are defined by Eq. \eqref{dj1} with $ \kappa_{l,j}=-l $ or $ (4-l)\ii^{j} $, and the large constant $ A=20 $, the NMR-type patterns of $ [4,3,2] $-order vector RW solution $ \mathbf{q}^{[4,3,2]}(x,t) $ for the $ 3 $-NLSE can be obtained, as shown in the first two rows of Fig. \ref{Fig3}. It can be observed that each of the NMR-type RW patterns consists of $ 20 $ first-order RWs, which respectively form a $ 180 $-degree sector and a jellyfish-like shape. 
	
	Furthermore, When the internal large parameters are $ d_{l,j}=\frac{1}{j}A^{j} $ $ (1\leq l \leq n, j\in \Lambda_{l}) $ with the large constant $ A=20 $, the MR-type patterns of $ [4,3,2] $-order vector RW solution $ \mathbf{q}^{[4,3,2]}(x,t) $ for the $ 3 $-NLSE are generated. As shown in Fig. \ref{Fig3} (vii-ix), each of these RW Patterns contains five first-order RWs and a lower-order RW of $ [3,3,2] $-order with the peaks forming a right double-arrow shape. In particular, this lower-order rogue wave is located at the origin of the $ (x,t) $-plane.

	\item[(2).] \textbf{Cases of the $ 4 $-NLSE}
	
	For the $ 4 $-NLSE, when the parameters $ a_{k} $ and $ b_{k} $ in the seed solution \eqref{seed} are
	\begin{equation}\label{abk3}
		\begin{aligned}
			&	(a_{1}, a_{2}, a_{3}, a_{4})= \left(\frac{2}{5}\, \csc \left( \frac{\pi}{5} \right), \frac{2}{5}\, \csc \left( \frac{2\pi}{5} \right), \frac{2}{5}\, \csc \left( \frac{2\pi}{5} \right), \frac{2}{5}\, \csc \left( \frac{\pi}{5} \right) \right) , \\
			&(b_{1}, b_{2}, b_{3}, b_{4})= \left(\frac{2}{5}\, \cot \left( \frac{\pi}{5} \right), \frac{2}{5}\, \cot \left( \frac{2\pi}{5} \right), -\frac{2}{5}\, \cot \left( \frac{2\pi}{5} \right), -\frac{2}{5}\, \cot \left( \frac{\pi}{5} \right) \right),
		\end{aligned}
	\end{equation}
    the characteristic equation \eqref{cheq} has a $ 5 $-multiple root $ \chi_{0}=\frac{2\ii}{5} $ at $ \lambda=\ii $. Meanwhile, we take the arbitrary parameters $ \lambda^{[1]}=-2\ii $ and $ \chi^{[1]}=-\frac{4}{5}\ii $. Then, when the internal large parameters $ \{d_{l,j}\}_{1\leq l \leq n, j\in \Lambda_{l}} $ are defined by Eq. \eqref{dj1} with $\kappa_{l,j}=$ $ l $ or $ \frac{l}{2}\ii^{j} $, and the large constant $ A=20 $, the NMR-type patterns of $ [3,2,2,1] $-order vector RW solution $ \mathbf{q}^{[3,2,2,1]}(x,t) $ for the $ 4 $-NLSE is obtained. Each of these NMR-type RW patterns consists of 14 first-order RWs, whose peaks form two distinct thumbtack-like shapes, as shown in Fig. \ref{Fig4} (i-iv) and (v-viii), respectively.

	Moreover, when the internal large parameters $ \{d_{l,j}\}_{1\leq l \leq n, j\in \Lambda_{l}} $ are defined by Eq. \eqref{dj2} with
	\begin{equation}\label{hatkap}
	\begin{aligned}
		&\kappa_{1,1}=-2, \quad \kappa_{l,j}= \frac{1}{j} , \quad j>1, \quad  A=20,\\
		&\{ \hat{\kappa}_{2}, \hat{\kappa}_{3}, \hat{\kappa}_{4}, \hat{\kappa}_{6}, \hat{\kappa}_{7}, \hat{\kappa}_{8}, \hat{\kappa}_{9} \}
		=\left\lbrace 3A, -3A, 3A, \frac{12}{5}A, -\frac{9}{5}A, \frac{6}{5}A, -\frac{3}{5}A\right\rbrace ,
	\end{aligned}
	\end{equation}
	the MR-type patterns of $ [3,2,2,1] $-order vector RW solution $ \mathbf{q}^{[3,2,2,1]}(x,t) $ for the $ 4 $-NLSE is constructed. As shown in Fig. \ref{Fig4} (ix-xii), each of the RW patterns contains four scattered first-order RWs and a lower-order RW of $ [2,2,2,1] $-order. The peaks in each pattern form a right-arrow shape, and the lower-order RW is located far from the origin of the $ (x,t) $-plane. 

\end{enumerate}

\begin{figure*}[!htbp]
	\centering
	\includegraphics[width=0.8\textwidth]{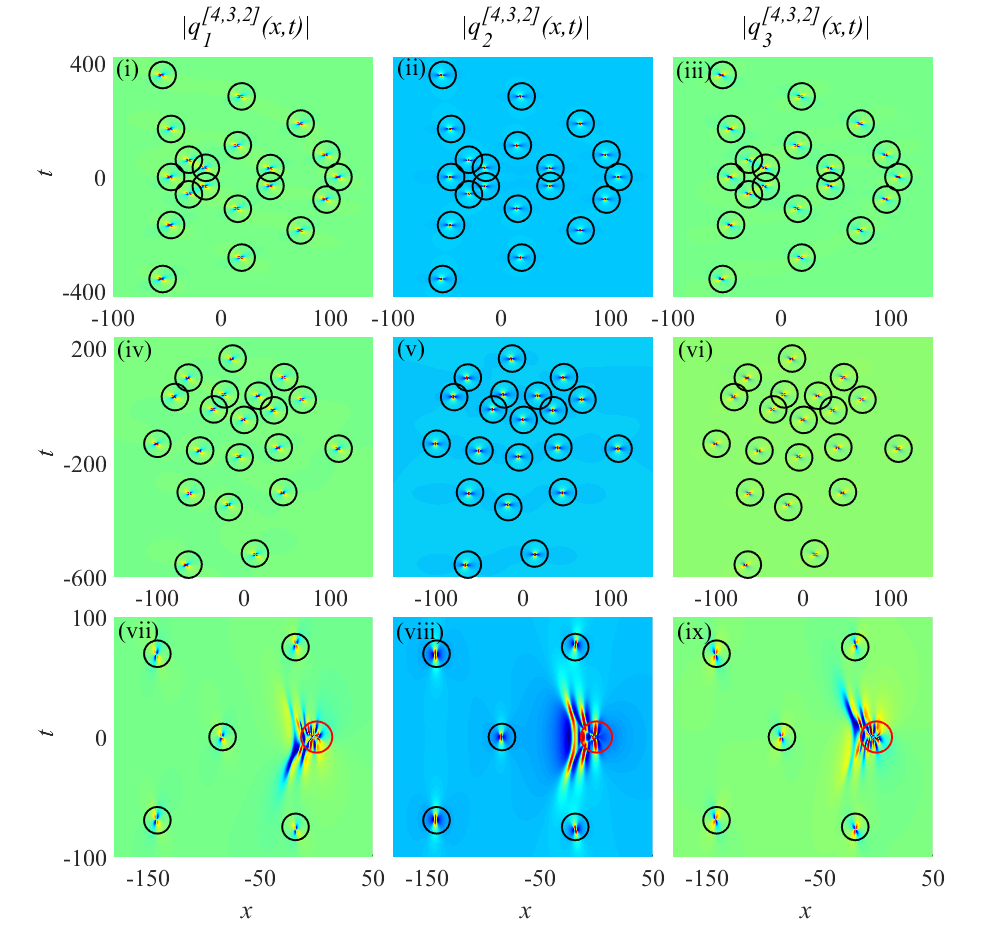}
	\caption{The $ [4,3,2] $-order vector RW solution $ \mathbf{q}^{[4,3,2]}(x,t) $ of the $ 3 $-NLSE is determined by the internal large parameters $ d_{l,j}=\kappa_{l,j}A^{j} $ with $ 1\leq l\leq n, j\in \Lambda_{l} $ and $ A=20 $, where $ \kappa_{l,j}=-l$ in (i-iii), $ \kappa_{l,j}=(4-l)\ii^{j} $ in (iv-vi), and $ \kappa_{l,j}=\frac{1}{j} $ in (vii-ix), respectively. From the first column to the third column, these are $ |{q}_{1}^{[4,3,2]}(x,t)| $, $ |{q}_{2}^{[4,3,2]}(x,t)| $, and $ |{q}_{3}^{[4,3,2]}(x,t)| $, respectively. These black circles represent the predicted position of each first-order RW in the simple-root region, and these red circles represent the predicted positions of the lower-order RWs in the multiple-root region.}
	\label{Fig3}
\end{figure*}

Based on Theorems \ref{theo_nmrp} and \ref{theo_mrp}, we also provide the predicted positions of all wave peaks of the RW patterns for the $ 3 $-NLSE in Fig. \ref{Fig3} and the $ 4 $-NLSE in Fig. \ref{Fig4}, respectively. The predicted positions in the simple-root region of the RW patterns are marked with black circles, while those in the multiple-root region are marked with red circles. Here, the NMR-type RW patterns can be regarded as having only a simple root region, without a multiple root region. The predicted positions of these peaks are asymptotically consistent with the wave peak positions in the true vector RW solutions. In addition, it is evident that the peak distributions of the RW Patterns in Figs. \ref{Fig3} and \ref{Fig4} are very similar to the root structures of the GMAM polynomials $ W_{[4,3,2]}(z) $ and $ W_{[3,2,2,1]}(z) $, as shown in Figs. \ref{Fig1} and \ref{Fig2}, respectively. They can be regarded as magnified and shifted versions of the latter to some extent.

\begin{figure*}[!htbp]
	\centering
	\includegraphics[width=\textwidth]{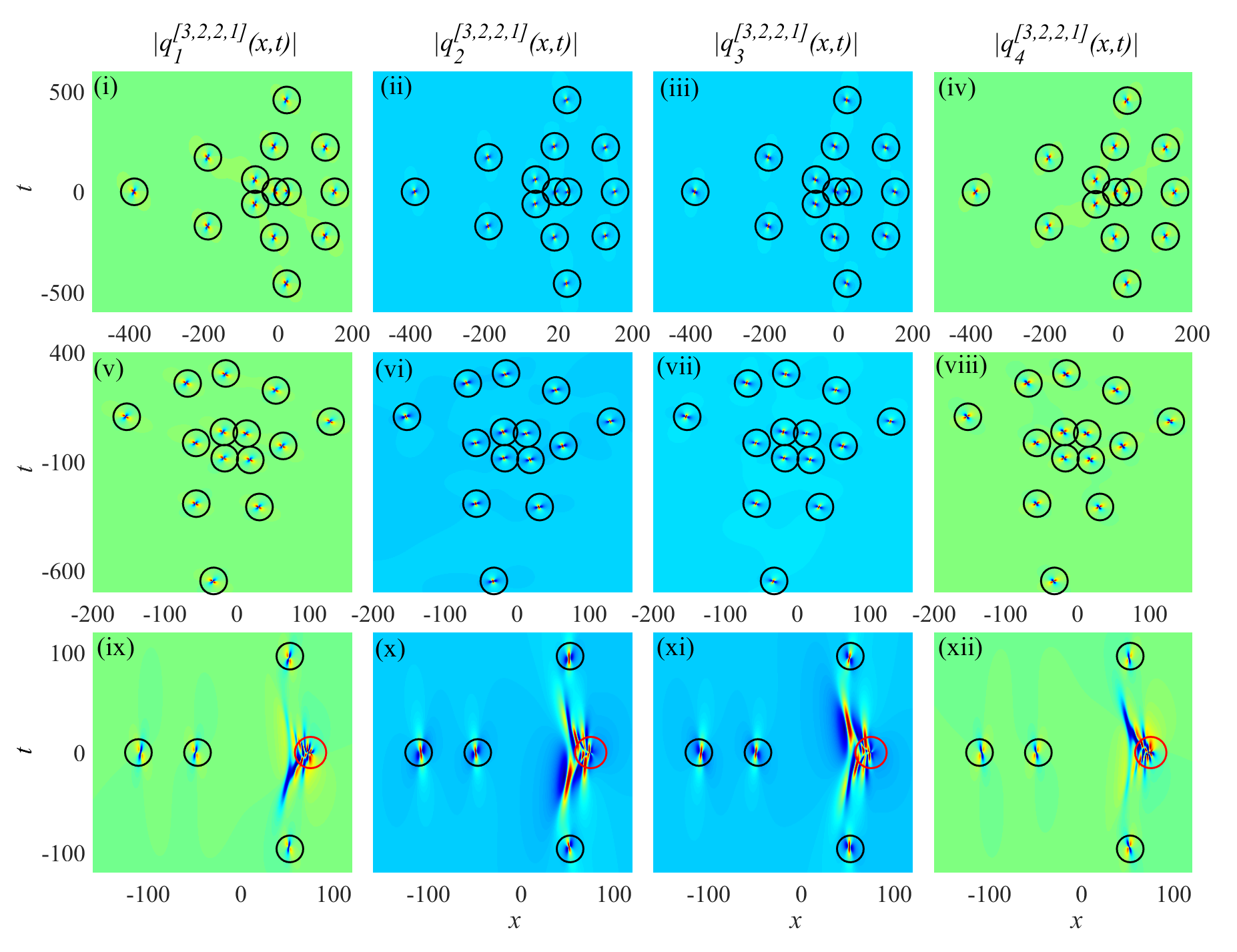}
	\caption{The $ [3,2,2,1] $-order vector RW solution $ \mathbf{q}^{[3,2,2,1]}(x,t) $ of the $ 4 $-NLSE is determined by the internal large parameters $ \{d_{l,j}\}_{1\leq l \leq n, j\in \Lambda_{l}} $ defined in Eq. \eqref{dj2} with $ A=20 $. The free parameters $ \kappa_{l,j} $ are given as follows: $ \kappa_{l,j}=l$ in (i-iiii), $ \kappa_{l,j}=\frac{l}{2}\ii^{j} $ in (v-viii), and $ \kappa_{l,j}$ is defined by Eq. \eqref{hatkap} in (vii-ix), respectively. From the first column to the fourth column, these are $ |{q}_{1}^{[3,2,2,1]}(x,t)| $, $ |{q}_{2}^{[3,2,2,1]}(x,t)| $, $ |{q}_{3}^{[3,2,2,1]}(x,t)| $, and $ |{q}_{4}^{[3,2,2,1]}(x,t)| $, respectively.
		These black circles represent the predicted position of each first-order RW in the simple-root region, and these red circles represent the predicted positions of the lower-order RWs in the multiple-root region.}
	\label{Fig4}
\end{figure*}	

\section{Proofs of the main results}\label{Sec_proof}

\subsection{Proof of Theorem \ref{theo_mrs}}\label{subsec_proofmrs}

This subsection presents the proof of Theorem \ref{theo_mrs}.

First, we will derive the multiplicity of the root $ z=z_{0}-\kappa_{1,1} $ for the GMAM polynomial $W_{\mathcal{N}}(z)$ when the integer index vector $ \mathcal{N}$ and the free parameters $ \kappa_{l,i} $ are defined by Theorem \ref{theo_mrs}. To achieve this goal, we apply the transformation $ z{\rightarrow} z + z_0-\kappa_{1,1} $, thereby reducing the problem to finding the multiplicity of the zero root of $W_{\mathcal{N}}(z+z_{0}-\kappa_{1,1})$. 

Moreover, for convenience, we denote $ \varphi_{l,j}(z) $ as $ \varphi_{j}(z) $ in the following discussion because of the definition \eqref{mmrp} of the free parameters $ \kappa_{l,i} $. Then, we introduce the new particular Schur polynomials $ \varphi_{j}(z+z_{0}-\kappa_{1,1};c) $ and the polynomials $W_{\mathcal{N}}(z+z_{0}-\kappa_{1,1};c)$, as follows:
\begin{equation}\label{npoly}
\begin{aligned}
	&\sum_{j=0}^{\infty} \varphi_{j}(z+z_{0}-\kappa_{1,1};c) \varepsilon^{j} = \exp\left( z\varepsilon+ \sum_{\substack{i=1 \\ (n+1) \nmid i}}^{\infty}\frac{ z_{0}^{i}}{i}c^{i}\varepsilon^{i} \right), \\
	&W_{\mathcal{N}}(z+z_{0}-\kappa_{1,1};c)=c_{\mathcal{N}} \mathcal{W}\left(\{\varphi_{(n+1)(j-1)+1}(z+z_{0}-\kappa_{1,1};c)\}_{j=1}^{N_{1}}, \{\varphi_{(n+1)(j-1)+2}(z+z_{0}-\kappa_{1,1};c)\}_{j=1}^{N_{2}},\right. \\ 
	& \qquad \qquad\qquad\qquad\qquad \left. \ldots, \{\varphi_{(n+1)(j-1)+n}(z+z_{0}-\kappa_{1,1};c)\}_{j=1}^{N_{n}} \right),
\end{aligned}
\end{equation}
where $ c $ is a constant, $ \varphi_{j}(z+z_{0}-\kappa_{1,1};c)=0 $ for $ j<0 $, and $ c_{\mathcal{N}} $ is defined by Eq. \eqref{gswhp}. Then, we can obtain
\begin{equation}\label{npoly2}
	\varphi_{j}(z+z_{0}-\kappa_{1,1};c) = c^{j} \varphi_{j}(\hat{z}+z_{0}-\kappa_{1,1}) , \quad
	W_{\mathcal{N}}(z+z_{0}-\kappa_{1,1};c)=c^{\Gamma}W_{\mathcal{N}}(\hat{z}+z_{0}-\kappa_{1,1})
\end{equation}
with $ \hat{z}=c^{-1}z $ and $ \Gamma $ given in Eq. \eqref{degpol}.

Assuming that each term of $ W_{\mathcal{N}}(z+z_{0}-\kappa_{1,1};c) $ takes the form $ c^{i}z^{j} $, we group the terms according to the power of $ z $ in the right equation in Eq. \eqref{npoly2}. For the exponent of $ c $ in the term $ \cO(z^{j}) $, we have
\begin{equation}\label{gij1}
	i=\Gamma-j,
\end{equation}
which indicates that the power of $ z $ in each term of $ W_{\mathcal{N}}(z+z_{0}-\kappa_{1,1};c) $ decreases as the power of $ c $ increases. Given that the multiplicity of the zero root in $W_{\mathcal{N}}(z+z_{0}-\kappa_{1,1})$ equals that of $W_{\mathcal{N}}(z+z_{0}-\kappa_{1,1};c)$, we can determine the multiplicity $ \Gamma_{0} $ of the zero root by considering the highest-power of $ c $ in $W_{\mathcal{N}}(z+z_{0}-\kappa_{1,1})$.

Next, we first consider the condition of \(N_1 > N_2 > \cdots > N_n > 0\) for the polynomial $W_{\mathcal{N}}(z+z_{0}-\kappa_{1,1};c)$. We reorder the columns of the Wronskian determinant in $ W_{\mathcal{N}}(z+z_{0}-\kappa_{1,1};c) $, arranging the subscripts of \(\varphi_{k}(z+z_{0}-\kappa_{1,1};c)\) in the first row in descending order. The transformed determinant can be expressed as
\begin{equation}\label{deter2}
	\mathcal{W}\left( H_{1}, H_{2}, \ldots, H_{n-1}, H_{n}\right), 
\end{equation}
where
\begin{equation}\label{deter2e}
\begin{aligned}
	& H_{1}= \left( \hat{\varphi}_{(n+1)(N_{1}-1)+1}, \hat{\varphi}_{(n+1)(N_{1}-2)+1}, \ldots, \hat{\varphi}_{(n+1)N_{2}+1}\right)_{1\times (N_{1}-N_{2})},\\
	& H_{2}= \left( \hat{\varphi}_{(n+1)(N_{2}-1)+2}, \hat{\varphi}_{(n+1)(N_{2}-1)+1}, \hat{\varphi}_{(n+1)(N_{2}-2)+2}, \hat{\varphi}_{(n+1)(N_{2}-2)+1}, \ldots, \hat{\varphi}_{(n+1)N_{3}+2}, \hat{\varphi}_{(n+1)N_{3}+1}\right)_{1\times 2(N_{2}-N_{3})},\\
	&\vdots\\
	& H_{n-1}= \left( \hat{\varphi}_{(n+1)(N_{n-1}-1)+n-1}, \hat{\varphi}_{(n+1)(N_{n-1}-1)+n-2}, \ldots, \hat{\varphi}_{(n+1)(N_{n-1}-1)+1}, \hat{\varphi}_{(n+1)(N_{n-1}-2)+n-1}, \hat{\varphi}_{(n+1)(N_{n-1}-2)+n-2},   \right. \\
	& \qquad\qquad \left. \ldots, \hat{\varphi}_{(n+1)(N_{n-1}-2)+1}, \ldots,\hat{\varphi}_{(n+1)N_{n}+n-1}, \hat{\varphi}_{(n+1)N_{n}+n-2}, \ldots, \hat{\varphi}_{(n+1)N_{n}+1} \right)_{1\times (n-1)(N_{n-1}-N_{n})},\\
	& H_{n}= \left( \hat{\varphi}_{(n+1)(N_{n}-1)+n}, \hat{\varphi}_{(n+1)(N_{n}-1)+n-1}, \ldots, \hat{\varphi}_{(n+1)(N_{n}-1)+1}, \hat{\varphi}_{(n+1)(N_{n}-2)+n}, \hat{\varphi}_{(n+1)(N_{n}-2)+n-1}, \ldots,  \right. \\
	& \qquad\quad \left. \hat{\varphi}_{(n+1)(N_{n}-2)+1}, \ldots,\hat{\varphi}_{n}, \hat{\varphi}_{n-1}, \ldots, \hat{\varphi}_{1} \right)_{1\times nN_{n}}
\end{aligned}
\end{equation}
with $ \hat{\varphi}_{k}={\varphi}_{k}(z+z_{0}-\kappa_{1,1};c) $. Here, we disregard the sign difference resulting from reordering the columns of the determinant because we are interested in powers of $ c $ and $z$ in the polynomial $W_{\mathcal{N}}(z+z_{0}-\kappa_{1,1};c)$.

Moreover, according to the definitions of $ \varphi_{j}(z+z_{0}-\kappa_{1,1}) $ and $ \varphi_{j}(z+z_{0}-\kappa_{1,1};c) $ in Eqs. \eqref{schpj} and \eqref{npoly}, we generate the following expansion:
\begin{equation}\label{spolyj1}
	\varphi_{j}(z+z_{0}-\kappa_{1,1};c) = \sum_{i=0}^{j} \varphi_{j-i}(z_{0}-\kappa_{1,1}) c^{j-i}\frac{z^{i}}{i!}.
\end{equation}
We expand all elements of the Wronskian determinant \eqref{deter2} by using the above formula \eqref{spolyj1}. Then, we perform a column transformation for the expanded determinant: employ the second column to eliminate the highest-order term of $ c $ in the first column. When $ N_{1}>N_{2}+1 $, the simplified first column is
\begin{equation}\label{deter3} 
	\begin{pmatrix}
		\sum_{k=0}^{(n+1)(N_{1}-2)+1} \left( \varphi_{(n+1)(N_{1}-1)+1-k} -\frac{\varphi_{(n+1)(N_{1}-1)+1}}{\varphi_{(n+1)(N_{1}-2)+1}} \varphi_{(n+1)(N_{1}-2)+1-k} \right) c^{(n+1)(N_{1}-1)+1-k} \frac{z^{k}}{k!} \\
		+\sum_{k=(n+1)(N_{1}-2)+2}^{(n+1)(N_{1}-1)+1} \varphi_{(n+1)(N_{1}-1)+1-k} c^{(n+1)(N_{1}-1)+1-k} \frac{z^{k}}{k!} \, ,\\
		\sum_{k=0}^{(n+1)(N_{1}-2)} \left( \varphi_{(n+1)(N_{1}-1)-k} -\frac{\varphi_{(n+1)(N_{1}-1)+1}}{\varphi_{(n+1)(N_{1}-2)+1}} \varphi_{(n+1)(N_{1}-2)-k} \right) c^{(n+1)(N_{1}-1)-k} \frac{z^{k}}{k!} \\
		+\sum_{k=(n+1)(N_{1}-2)+1}^{(n+1)(N_{1}-1)} \varphi_{(n+1)(N_{1}-1)-k} c^{(n+1)(N_{1}-1)-k} \frac{z^{k}}{k!} \, ,\\
		\vdots
	\end{pmatrix}
\end{equation}
with $\varphi_{k}= \varphi_{k}(z_{0}-\kappa_{1,1}) $. From Proposition \ref{prop1} in the Appendix \ref{Appe_prop1}, we find that the first column \eqref{deter3} {is} equal to
\begin{equation}\label{deter3b}
	\left( \varphi_{n}c^{n} \frac{z^{(n+1)(N_{1}-2)+2}}{[(n+1)(N_{1}-2)+2]!} +\cO(c^{n-1}), \varphi_{n}c^{n} \frac{z^{(n+1)(N_{1}-2)+1}}{[(n+1)(N_{1}-2)+1]!} +\cO(c^{n-1}), \ldots\right)^{T}.
\end{equation}

When $ N_{1}=N_{2}+1 $, the simplified first column in the expanded determinant is
\begin{equation}\label{deter3c} 
	\begin{pmatrix}
		\sum_{k=0}^{(n+1)(N_{2}-1)+2} \left( \varphi_{(n+1)(N_{1}-1)+1-k} -\frac{\varphi_{(n+1)(N_{1}-1)+1}}{\varphi_{(n+1)(N_{2}-1)+2}} \varphi_{(n+1)(N_{2}-1)+2-k} \right) c^{(n+1)(N_{1}-1)+2-k} \frac{z^{k}}{k!} \\
		+\sum_{k=(n+1)(N_{2}-1)+3}^{(n+1)(N_{1}-1)+1} \varphi_{(n+1)(N_{1}-1)+1-k} c^{(n+1)(N_{1}-1)+1-k} \frac{z^{k}}{k!} \, ,\\
		\sum_{k=0}^{(n+1)(N_{2}-1)+1} \left( \varphi_{(n+1)(N_{1}-1)-k} -\frac{\varphi_{(n+1)(N_{1}-1)+1}}{\varphi_{(n+1)(N_{2}-1)+2}} \varphi_{(n+1)(N_{2}-1)+1-k} \right) c^{(n+1)(N_{1}-1)-k} \frac{z^{k}}{k!} \\
		+\sum_{k=(n+1)(N_{2}-1)+2}^{(n+1)(N_{1}-1)} \varphi_{(n+1)(N_{1}-1)-k} c^{(n+1)(N_{1}-1)-k} \frac{z^{k}}{k!} \, ,\\
		\vdots
	\end{pmatrix}.
\end{equation} 
Now, by Proposition \ref{prop1} in the Appendix \ref{Appe_prop1}, we determine that the first column \eqref{deter3c} {is} equal to
\begin{equation}\label{deter3d}
	\left( \varphi_{n-1}c^{n-1} \frac{z^{(n+1)(N_{2}-1)+3}}{[(n+1)(N_{2}-1)+3]!} +\cO(c^{n-2}), \varphi_{n-1}c^{n-1} \frac{z^{(n+1)(N_{2}-1)+2}}{[(n+1)(N_{2}-1)+2]!} +\cO(c^{n-2}), \ldots\right)^{T}.
\end{equation}

Similarly, we continue performing \(N-2\) column transformations on the expanded determinant: starting with the second column, we successively use the elements of the following column to eliminate the highest power terms of $ c $ in the preceding column. The detailed calculation process is omitted here.

Therefore, after expansion and column transformations, the Wronskian determinant \eqref{deter2} can be simplified to
\begin{equation}\label{deter4}
	\det\left( \mathbf{H}_{1}, \mathbf{H}_{2}, \ldots, \mathbf{H}_{n} \right), 
\end{equation}
where

\[\mathbf{H}_{1}=
	\begin{pmatrix}
		c^{n}\varphi_{n} \frac{z^{(n+1)(N_{1}-2)+2}}{[(n+1)(N_{1}-2)+2]!} +\cO(c^{n-1}) & c^{n}\varphi_{n} \frac{z^{(n+1)(N_{1}-2)+1}}{[(n+1)(N_{1}-2)+1]!} +\cO(c^{n-1}) & \cdots\\
		c^{n}\varphi_{n} \frac{z^{(n+1)(N_{1}-3)+2}}{[(n+1)(N_{1}-3)+2]!} +\cO(c^{n-1}) & c^{n}\varphi_{n} \frac{z^{(n+1)(N_{1}-3)+1}}{[(n+1)(N_{1}-3)+1]!} +\cO(c^{n-1}) & \cdots\\
		\vdots & \vdots & \ddots \\
		c^{n}\varphi_{n} \frac{z^{(n+1)N_{2}+2}}{[(n+1)N_{2}+2]!} +\cO(c^{n-1}) & c^{n}\varphi_{n} \frac{z^{(n+1)N_{2}+1}}{[(n+1)N_{2}+1]!} +\cO(c^{n-1}) & \cdots\\
		c^{n-1}\varphi_{n-1} \frac{z^{(n+1)(N_{2}-1)+3}}{[(n+1)(N_{2}-1)+3]!} +\cO(c^{n-2}) & c^{n-1}\varphi_{n-1} \frac{z^{(n+1)(N_{2}-1)+2}}{[(n+1)(N_{2}-1)+2]!} +\cO(c^{n-2}) & \cdots
	\end{pmatrix}_{(N_{1}-N_{2})\times N}^{T},\]\\
\[\mathbf{H}_{2}=
	\begin{pmatrix}
		\frac{z^{(n+1)(N_{2}-1)+2}}{[(n+1)(N_{2}-1)+2]!} & \frac{z^{(n+1)(N_{2}-1)+1}}{[(n+1)(N_{2}-1)+1]!} & \cdots\\
		c^{n-1}\varphi_{n-1} \frac{z^{(n+1)(N_{2}-2)+3}}{[(n+1)(N_{2}-2)+3]!} +\cO(c^{n-2}) & c^{n-1}\varphi_{n-1} \frac{z^{(n+1)(N_{2}-2)+2}}{[(n+1)(N_{2}-2)+2]!} +\cO(c^{n-2}) & \cdots \\
		\frac{z^{(n+1)(N_{2}-2)+2}}{[(n+1)(N_{2}-2)+2]!} & \frac{z^{(n+1)(N_{2}-2)+1}}{[(n+1)(N_{2}-2)+1]!} & \cdots\\
		c^{n-1}\varphi_{n-1} \frac{z^{(n+1)(N_{2}-3)+3}}{[(n+1)(N_{2}-3)+3]!} +\cO(c^{n-2}) & c^{n-1}\varphi_{n-1} \frac{z^{(n+1)(N_{2}-3)+2}}{[(n+1)(N_{2}-3)+2]!} +\cO(c^{n-2}) & \cdots \\
			\vdots & \vdots & \ddots \\
		\frac{z^{(n+1)N_{3}+2}}{[(n+1)N_{3}+2]!} & \frac{z^{(n+1)N_{3}+1}}{[(n+1)N_{3}+1]!} & \cdots\\
		c^{n-2}\varphi_{n-2} \frac{z^{(n+1)(N_{3}-1)+4}}{[(n+1)(N_{3}-1)+4]!} +\cO(c^{n-3}) & c^{n-2}\varphi_{n-2} \frac{z^{(n+1)(N_{3}-1)+3}}{[(n+1)(N_{3}-1)+3]!} +\cO(c^{n-3}) & \cdots
	\end{pmatrix}_{2(N_{2}-N_{3})\times N}^{T},\]\\

\hspace*{3em}{\vdots}

\[\mathbf{H}_{n}=
	\begin{pmatrix}
		\frac{z^{(n+1)(N_{n}-1)+n}}{[(n+1)(N_{n}-1)+n]!} & \frac{z^{(n+1)(N_{n}-1)+n-1}}{[(n+1)(N_{n}-1)+n-1]!} & \cdots\\
			\vdots & \vdots & \ddots \\
		\frac{z^{(n+1)(N_{n}-1)+2}}{[(n+1)(N_{n}-1)+2]!} & \frac{z^{(n+1)(N_{n}-1)+1}}{[(n+1)(N_{n}-1)+1]!} & \cdots\\
		c \varphi_{1} \frac{z^{(n+1)(N_{n}-1)}}{[(n+1)(N_{n}-1)]!} +\cO(c^0) & c \varphi_{1} \frac{z^{(n+1)(N_{n}-1)}-1}{[(n+1)(N_{n}-1)-1]!} +\cO(c^0) & \cdots\\
		\frac{z^{(n+1)(N_{n}-2)+n}}{[(n+1)(N_{n}-2)+n]!} & \frac{z^{(n+1)(N_{n}-2)+n-1}}{[(n+1)(N_{n}-2)+n-1]!} & \cdots\\
			\vdots & \vdots & \ddots \\
		\frac{z^{(n+1)(N_{n}-2)+2}}{[(n+1)(N_{n}-2)+2]!} & \frac{z^{(n+1)(N_{n}-2)+1}}{[(n+1)(N_{n}-2)+1]!} & \cdots\\
		c \varphi_{1} \frac{z^{(n+1)(N_{n}-2)}}{[(n+1)(N_{n}-2)]!} +\cO(c^0) & c \varphi_{1} \frac{z^{(n+1)(N_{n}-2)}-1}{[(n+1)(N_{n}-2)-1]!} +\cO(c^0) & \cdots\\
			\vdots & \vdots & \ddots \\
		\frac{z^{n}}{n!} & \frac{z^{n-1}}{(n-1)!} & \cdots\\
			\vdots & \vdots & \ddots \\
		\frac{z^{2}}{2!} & z & \cdots\\
		c \varphi_{1} +z & 1 & \cdots
	\end{pmatrix}_{nN_{n}\times N}^{T}.\]

Then, we obtain the highest-order term of $ c $ for the determinant \eqref{deter4} below:
\begin{equation}\label{deter5}
	C_{2}c^{\Gamma_{1}} \mathcal{W}\left( \bar{H}_{0}, \bar{H}_{2}, \bar{H}_{3}, \ldots, \bar{H}_{n}\right), 
\end{equation}
where $ \Gamma_{1}=(n+1)N_{1}-N-n+1 $, $ C_{2} $ is nonzero constant, and
\begin{equation}\label{deter5b}
\begin{aligned}
	&\bar{H}_{l}=\left( z^{l}, z^{n+1+l}, \ldots, z^{(n+1)(\bar{N}_{l}-1)+l}\right), \quad l=0,2,3,\ldots, n,\\
	&\bar{N}_{0}=N_{n}+1, \quad \bar{N}_{2}=N_{1}-1, \quad \bar{N}_{j}=N_{j-1}, \quad 3\leq j \leq n.
\end{aligned}
\end{equation}
We expand the Wronskian determinant \eqref{deter5} along the first column as follows:
\begin{equation}\label{deter6}
	\begin{aligned}
		&C_{2}c^{\Gamma_{1}} \mathcal{W}\left( \hat{H}_{n}, \hat{H}_{1}, \hat{H}_{2}, \ldots, \hat{H}_{n-1}\right),\\ &\hat{H}_{l}=\left( z^{l}, z^{n+1+l}, \ldots, z^{(n+1)(\hat{N}_{l}-1)+l}\right), \quad l=1,2,\ldots, n,\\
	\end{aligned}
\end{equation}
where $ \hat{N}_{l} $ $ (1\leq l\leq n) $ are given in Eq. \eqref{gamma0} with $ N_{1}>N_{2}>\cdots>N_{n}>0 $. Therefore, we conclude that when $ N_{1}>N_{2}>\cdots>N_{n}>0 $, the power $ \Gamma_{0} $ of $ z $ in the highest-order term of $ c $ for $ W_{\mathcal{N}}(z+z_{0}-\kappa_{1,1};c) $ is given by Eq. \eqref{degpol}, where $ N_{l} $ $ (1\leq l\leq n) $ are replaced by $ \hat{N}_{l} $ defined in Eq. \eqref{gamma0}. 

In general, when $ N_1 \geq N_2 \geq \cdots \geq N_n \geq 0 $, the determinant \eqref{deter2} does not exist the parts of $ H_{k} $ $ (i\leq k<j) $ for $ N_{i}=N_{j} $ $ (1\leq i<j\leq n) $ and the part of $ H_{n} $ for $ N_{n}=0 $. Then, we can also obtain the power $ \Gamma_{0} $ of $ z $ in the highest-order term of $ c $ for $ W_{\mathcal{N}}(z+z_{0}-\kappa_{1,1};c) $ by performing calculations similar to steps \eqref{spolyj1}-\eqref{deter6} above. Here, $ \Gamma_{0} $ is also defined by Eq. \eqref{degpol}, where $ N_{l} $ $ (1\leq l\leq n) $ are replaced by $ \hat{N}_{l} $ given in Eq. \eqref{gamma0}. Therefore, we derive that the multiplicity of the zero root in $ W_{\mathcal{N}}(z+z_{0}-\kappa_{1,1}) $ (alias the multiplicity of the root $ z=z_{0}-\kappa_{1,1} $ of $ W_{\mathcal{N}}(z) $) is $ \Gamma_{0} $. This completes the proof of Theorem \ref{theo_mrs}.

\subsection{Proof of Theorem \ref{theo_nmrp}}\label{subsec_proofnmrp}

This subsection provides the detailed proof of Theorem \ref{theo_nmrp}.

For the formula \eqref{horw} of the $\mathcal{N}$-order vector RW solution $\mathbf{q}^{\mathcal{N}}(x,t)$, we first rewrite the determinant of $ \mathbf{M}^{[s]} $ as
\begin{equation}\label{msij1}
    \begin{aligned}
    &\det\left(\mathbf{M}^{[s]}\right)=\begin{vmatrix}
        \mathbf{0}_{N\times N} & -\mathbf{M}^{[s,-]} \\
        \mathbf{M}^{[s,+]} & \mathbb{I}_{(n+1)N_{max}}
    \end{vmatrix}, \quad 0 \leq s \leq n,
    \end{aligned}
\end{equation}
where $\mathbb{I}_{(n+1)N_{max}}$ is an $ (n+1)N_{max}\times (n+1)N_{max} $ identity matrix with $N_{max}= \max(N_{1},N_{2}, \ldots, N_{n})$, and the matrices $\mathbf{M}^{[s,\pm]}$ are
\begin{equation}\label{msij1b}
\begin{aligned}
    &\mathbf{M}^{[s,+]}= \left( M^{[s,+]}_{1}, M^{[s,+]}_{2}, \ldots, M^{[s,+]}_{n} \right), \quad \mathbf{M}^{[s,-]}= \left( {M^{[s,-]}_{1}}^{T}, {M^{[s,-]}_{2}}^{T}, \ldots, {M^{[s,-]}_{n}}^{T} \right)^{T},\\
    & M^{[s,+]}_{l}=\left( C_{1}^{(i-1)/2} S_{(n+1)(j-1)+l-i+1}(\mathbf{u}^{[s,l,+]}(i-1)) \right)_{1\leq i\leq (n+1)N_{max}, 1\leq j\leq N_{l}}, \\  
    &M^{[s,-]}_{l}=\left( C_{1}^{(j-1)/2} S_{(n+1)(i-1)+l-j+1}(\mathbf{u}^{[s,l,-]}(j-1)) \right)_{1\leq i\leq N_{l}, 1\leq j\leq (n+1)N_{max}}, \quad 1\leq l \leq n,
\end{aligned}
\end{equation}
$N_{l}$ and $C_{1}$ are given in Theorem \ref{Theo1}, the vectors $\mathbf{u}^{[s,l,\pm]}(v)=( u^{[s,l,\pm]}_{1}(v), u^{[s,l,\pm]}_{2}(v), \ldots)$ are defined by
\begin{equation}\label{vuspm}
    \begin{aligned}
    &\mathbf{u}^{[0,l,+]}(v)=\mathbf{x}^{[l,+]} +v\mathbf{h}_{1}, \quad 
    \mathbf{u}^{[0,l,-]}(v)=\mathbf{x}^{[l,-]} +v\mathbf{h}_{1}^*,\\
    &\mathbf{u}^{[k,l,+]}(v)=\mathbf{x}^{[l,+]} +v\mathbf{h}_{1} -{\mathbf{h}_{2}^{[k]}}, \quad
    \mathbf{u}^{[k,l,-]}(v)=\mathbf{x}^{[l,-]} +v\mathbf{h}_{1}^* +({\mathbf{h}_{2}^{[k]}})^* , \quad 1\leq k \leq n,
    \end{aligned}
\end{equation}
and $\mathbf{x}^{[l,\pm]}, \mathbf{h}_{1}$, and ${\mathbf{h}_{2}^{[k]}}$ are given in Eqs. \eqref{xpm1} and \eqref{xpm2}.

Then, by utilizing the Cauchy--Binet formula, the determinant \eqref{msij1} can be further rewritten as
\begin{equation}\label{msij2}
\begin{aligned}
    &\det\left(\mathbf{M}^{[s]}\right)= \sum_{0\leq v_{1}<v_{2}<\cdots<v_{N}\leq (n+1)N_{max}-1} \left[ \det_{1\leq i,j \leq N}\left( \hat{\mathbf{M}}^{[s,+]}(\mathbf{v}) \right)\det_{1\leq i,j \leq N}\left( \hat{\mathbf{M}}^{[s,-]}(\mathbf{v}) \right) \right],\\
    &\hat{\mathbf{M}}^{[s,+]}(\mathbf{v})=\left( \hat{M}^{[s,+]}_{1}(\mathbf{v}), \hat{M}^{[s,+]}_{2}(\mathbf{v}), \ldots, \hat{M}^{[s,+]}_{n}(\mathbf{v}) \right), \\ 
    &\hat{\mathbf{M}}^{[s,-]}(\mathbf{v})=\left( (\hat{M}^{[s,-]}_{1}(\mathbf{v}))^{T}, ({\hat{M}^{[s,-]}_{2}}(\mathbf{v}))^{T}, \ldots, ({\hat{M}^{[s,-]}_{n}}(\mathbf{v}))^{T} \right)^{T},\\
    &\hat{M}^{[s,+]}_{l}(\mathbf{v})=\left( C_{1}^{v_{i}/2} S_{(n+1)(j-1)+l-v_{i}}(\mathbf{u}^{[s,l,+]}(v_{i})) \right)_{1\leq i\leq N, 1\leq j\leq N_{l}}, \\
    &\hat{M}^{[s,-]}_{l}(\mathbf{v})=\left( C_{1}^{v_{j}/2} S_{(n+1)(i-1)+l-v_{j}}(\mathbf{u}^{[s,l,-]}(v_{j})) \right)_{1\leq i\leq N_{l}, 1\leq j\leq N}, \quad 0 \leq s\leq n
\end{aligned}
\end{equation}
with the integer index vector $\mathbf{v}=(v_{1}, v_{2}, \ldots, v_{N})$. 

When the internal large parameters $\{d_{l,j} \}_{1\leq l\leq n, j\in \Lambda_{l}}$ in $\mathbf{q}^{\mathcal{N}}(x,t)$ are defined by Eq. \eqref{dj1} with $A\gg1$, we can derive
\begin{equation}\label{smasy1}
\begin{aligned}
    &S_{j}\left( \mathbf{u}^{[s,l,+]}(v) \right)
    = S_{j}\left( u^{[s,l,+]}_{1}(v), u^{[s,l,+]}_{2}(v), \ldots  \right)  
    =   S_{j}\left(x_{1}^{+}+\kappa_{l,1}A, \kappa_{l,2}A^{2}, \kappa_{l,3}A^{3}, \ldots \right) \left[1+\cO(A^{-1})\right], \\
    &S_{j}\left(x_{1}^{+}+\kappa_{l,1}A, \kappa_{l,2}A^{2}, \kappa_{l,3}A^{3}, \ldots \right)= A^{j} \varphi_{l,j}(z), \quad z=x_{1}^{+}A^{-1}=\ii \chi^{[1]} (x+ \chi_{0} t) A^{-1}, 
\end{aligned}
\end{equation}
based on definition \ref{schur} of the Schur polynomials, where $ \chi_{0}$, $\chi^{[1]} $ and $\varphi_{j}(z)$ are defined by Eqs. \eqref{chi0}, \eqref{lamchi}, and \eqref{schpj}, respectively.

Moreover, since the order of $A$ in the Schur polynomials $S_{j}(\mathbf{u}^{[s,l, \pm]}(v))$ decreases as the subscript $j$ decreases, we consider the first two highest-order terms of $A$ in Eq. \eqref{msij2} by two different indexes to analyze the asymptotics of $\mathbf{q}^{\mathcal{N}}(x,t)$ with $A\gg1$, where one is $\mathbf{v}_{1}=(0,1, \ldots, N-2, N-1)$ and another is $\mathbf{v}_{2}=(0,1, \ldots, N-2, N)$.

If choosing the first index $\mathbf{v}_{1}=(0,1, \ldots, N-2, N-1)$ in Eq. \eqref{msij2}, we obtain the highest-order term of $A$, as follows:
\begin{equation}\label{hoterm1}
    \det_{1\leq i,j \leq N}\left( \hat{\mathbf{M}}^{[s,+]}(\mathbf{v}_{1}) \right) \det_{1\leq i,j \leq N}\left( \hat{\mathbf{M}}^{[s,-]}(\mathbf{v}_{1}) \right).
\end{equation}
As $A\gg1$, we employ the above formula \eqref{smasy1} to generate the asymptotic expression
\begin{equation}\label{hmspasy1}
     \det_{1\leq i,j \leq N}\left( \hat{\mathbf{M}}^{[s,+]}(\mathbf{v}_{1}) \right) = C_{1}^{N(N-1)/2}c_{\mathcal{N}}^{-1}A^{\Gamma} W_{\mathcal{N}}(z) \left[ 1+\cO(A^{-1})\right],
\end{equation}
where $\Gamma$ is given in Eq. \eqref{degpol}, the GMAM polynomial $W_{\mathcal{N}}(z)$ and the constant $c_{\mathcal{N}}$ are defined by \eqref{gswhp}. Similarly, we also obtain
\begin{equation}\label{hmspasy1b}
     \det_{1\leq i,j \leq N}\left( \hat{\mathbf{M}}^{[s,-]}(\mathbf{v}_{1}) \right) = C_{1}^{N(N-1)/2}c_{\mathcal{N}}^{-1}A^{\Gamma} \left(W_{\mathcal{N}}(z)\right)^{*} \left[ 1+\cO(A^{-1})\right].
\end{equation}
Then, according to the highest-order term \eqref{hoterm1} of $A$ in Eq. \eqref{msij2}, we derive
\begin{equation}\label{asyhot1}
    \det\left(\mathbf{M}^{[s]}\right) = C_{1}^{N(N-1)}c_{\mathcal{N}}^{-2}  A^{2\Gamma} \left| W_{\mathcal{N}}(z) \right|^2 \left[ 1+\cO(A^{-1})\right].
\end{equation}

Substituting the above approximated expression \eqref{asyhot1} into the formula \eqref{horw} of the vector RW solution $ \mathbf{q}^{\mathcal{N}}(x,t) $, we find that $ q_{k}^{\mathcal{N}}(x,t) $ $(1\leq k \leq n)$ tend to the plane wave background $ q^{[0]}_{k}(x,t) $ as $ A\rightarrow \infty $, except for some locations at or near $ (\bar{x}_{0}, \bar{t}_{0}) $ in the specific region with $\sqrt{x^{2}+t^{2}} = \cO(A) $. Here, $ q^{[0]}_{k}(x,t) $ is given in Eq. \eqref{seed} and $ (\bar{x}_{0}, \bar{t}_{0}) $ satisfies
\begin{equation}\label{zs0}
    {z}_{0}=\ii \chi^{[1]} (\bar{x}_{0}+ \chi_{0} \bar{t}_{0}) A^{-1},
\end{equation}
with $\bar{z}_{0}$ being the root of $W_{\mathcal{N}}(z)$. Based on the assumption in Theorem \ref{theo_nmrp}, the root $\bar{z}_{0}$ of the GMAM polynomial $W_{\mathcal{N}}(z)$ is simple. 

Next, we consider the asymptotics of $ \mathbf{q}^{\mathcal{N}}(x,t) $ near the position $ (\bar{x}_{0}, \bar{t}_{0}) $ of the $(x,t)$-plane. Taking a coordinate transformation
\begin{equation}\label{hatxt1}
    \hat{x}=x-\hat{x}_{0}A, \quad \hat{t}=t-\hat{t}_{0}A,
\end{equation}
we have 
\begin{equation}\label{usi1}
    u_{i}^{[s,l,+]}(x,t)= u_{i}^{[s,l,+]}(\hat{x},\hat{t}) +(\alpha_{i}\hat{x}_{0} +\beta_{i}\ii \hat{t}_{0})A, \quad i\geq 1
\end{equation}
with $(\hat{x}_{0}, \hat{t}_{0})= (\bar{x}_{0}, \bar{t}_{0})A^{-1}$, and $\alpha_{i}$ and $\beta_{i}$ given in Eq. \eqref{xpm2}. Then, we derive
\begin{equation}\label{smasy3}
\begin{aligned}
    &\sum_{j=0}^{\infty}S_{j}(\mathbf{u}^{[s,l,+]}(v))(A^{-1}\varepsilon)^{j}\\
    =& \exp\left( {u_{1}^{[s,l,+]}({x},{t})} A^{-1}\varepsilon +{u_{2}^{[s,l,+]}({x},{t})} (A^{-1}\varepsilon)^{2}+ \cdots \right)\\
    =& \exp\left( \left( (\alpha_{1}\hat{x}_{0} +\beta_{1}\ii \hat{t}_{0})+ u_{1}^{[s,l,+]}(\hat{x},\hat{t})A^{-1} \right)\varepsilon +\left( (\alpha_{2}\hat{x}_{0} +\beta_{2}\ii \hat{t}_{0}) + u_{2}^{[s,l,+]}(\hat{x},\hat{t})A^{-1} \right)A^{-1}\varepsilon^2 +\cdots \right)\\
    =&\exp\left( (\alpha_{1}\hat{x}_{0} +\beta_{1}\ii \hat{t}_{0})\varepsilon +\sum_{\substack{i=1 \\ (n+1) \nmid i}}^{\infty} \kappa_{l,i}\varepsilon^i \right) 
    \exp\left( u_{1}^{[s,+]}(\hat{x},\hat{t})A^{-1}\varepsilon + (\alpha_{2}\hat{x}_{0} +\beta_{2}\ii \hat{t}_{0}) A^{-1}\varepsilon^2 + \cO(A^{-2})\right) \\
    =&\left(\sum_{j=0}^{\infty} \varphi_{l,j}(\bar{z}_{0})\varepsilon^j  \right) \left( 1+ (u_{1}^{[s,+]}(\hat{x},\hat{t})\varepsilon + (\alpha_{2}\hat{x}_{0} +\beta_{2}\ii \hat{t}_{0})\varepsilon^2 ) A^{-1} + \cO(A^{-2}) \right), 
\end{aligned}
\end{equation}
where $ u_{1}^{[s,+]}(\hat{x},\hat{t})=u_{1}^{[s,l,+]}(\hat{x},\hat{t})- \kappa_{l,1}A $, $ \alpha_{1}=\ii\chi^{[1]} $, and $ \beta_{1}=\chi^{[1]}\chi_{0} $. Thus, we obtain
\begin{equation}\label{smasy4}
    S_{j}(\mathbf{u}^{[s,l,+]}(v)) = A^j \left[ \varphi_{l,j}(\bar{z}_{0}) +A^{-1} \left(u_{1}^{[s,+]}(\hat{x},\hat{t}) \varphi_{l,j-1}(\bar{z}_{0}) + (\alpha_{2}\hat{x}_{0} +\beta_{2}\ii \hat{t}_{0})\varphi_{l,j-2}(\bar{z}_{0}) \right) \right] \left[ 1+\cO(A^{-2})\right],
\end{equation}
and similarly generate
\begin{equation}\label{smasy5}
    S_{j}(\mathbf{u}^{[s,l,-]}(v)) = A^j \left[ \varphi_{l,j}^*(\bar{z}_{0}) +A^{-1} \left(u_{1}^{[s,-]}(\hat{x},\hat{t}) \varphi_{l,j-1}^*(\bar{z}_{0}) + (\alpha_{2}^*\hat{x}_{0} -\beta_{2}^*\ii \hat{t}_{0})\varphi_{l,j-2}^*(\bar{z}_{0}) \right) \right] \left[ 1+\cO(A^{-2})\right]
\end{equation}
with $ u_{1}^{[s,-]}(\hat{x},\hat{t})=u_{1}^{[s,l,-]}(\hat{x},\hat{t})- \kappa_{l,1}^{*}A $.

Now, for the first index $ \mathbf{v}_{1}=(0,1, \ldots, N-2,N-1) $, we calculate
\begin{equation}\label{hmspasy2}
    \begin{aligned}
     \det_{1\leq i,j \leq N}\left( \hat{\mathbf{M}}^{[s,+]}(\mathbf{v}_{1}) \right) =& c_{\mathcal{N}}^{-1}C_{1}^{N(N-1)/2}A^{\Gamma-1} \left[ u_{1}^{[s,+]}(\hat{x},\hat{t}) W_{\mathcal{N}}'(\bar{z}_{0}) + (\alpha_{2}\hat{x}_{0} +\beta_{2}\ii \hat{t}_{0}) W_{\mathcal{N}}^{[2]}(\bar{z}_{0}) \right] \left[ 1+\cO(A^{-1})\right],\\
     \det_{1\leq i,j \leq N}\left( \hat{\mathbf{M}}^{[s,-]}(\mathbf{v}_{1}) \right)
    =& c_{\mathcal{N}}^{-1}C_{1}^{N(N-1)/2}A^{\Gamma-1} \left[ u_{1}^{[s,-]}(\hat{x},\hat{t}) (W_{\mathcal{N}}'(\bar{z}_{0}))^* + (\alpha_{2}^*\hat{x}_{0} -\beta_{2}^*\ii \hat{t}_{0}) \left(W_{\mathcal{N}}^{[2]}(\bar{z}_{0})\right)^* \right]  \\
    &\times \left[ 1+\cO(A^{-1})\right],
    \end{aligned}
\end{equation}
based on the above expressions \eqref{smasy4}-\eqref{smasy5} and $W_{\mathcal{N}}(\bar{z}_{0})=0 $.
Here, the polynomial $W_{\mathcal{N}}^{[2]}(\bar{z}_{0})$ represents a summation of two variations of $ W_{\mathcal{N}}(\bar{z}_{0})$, as follows: 
\begin{equation}\label{wnp2}
\begin{aligned}
    & W_{\mathcal{N}}^{[2]}(z_{0}) = c_{\mathcal{N}} \left[ 
    \det_{1\leq i,j \leq N}\left( \varphi_{l_{j},(n+1)p_{j}+l_{j}-\bar{p}_{i}}(\bar{z}_{0}) \right) + \det_{1\leq i,j \leq N}\left( \varphi_{l_{j},(n+1)p_{j}+l_{j}-\hat{p}_{i}}(\bar{z}_{0}) \right)
    \right],\\
    &\begin{array}{ll}
        l_{j}=\left\{\begin{array}{ll}
        1, & 1\leq j\leq N_{1}, \\
        2, & N_{1}< j\leq N_{1}+N_{2}, \\
        \ldots & \\
        n, & \sum_{k=1}^{n-1}N_{k} < j\leq N,
    \end{array} \right. & p_{j}=\left\{\begin{array}{ll}
        j-1, & 1\leq j\leq N_{1}, \\
        j-N_{1}-1, & N_{1}< j\leq N_{1}+N_{2}, \\
        \ldots & \\
        j-\sum_{k=1}^{n-1}N_{k} -1, & \sum_{k=1}^{n-1}N_{k} < j\leq N,
    \end{array}  \right. \\
       \bar{p}_{i}=\left\{\begin{array}{ll}
        N, & i=N-1, \\
        i-1, & i\ne N-1,
    \end{array}  \right.  & \hat{p}_{i}=\left\{\begin{array}{ll}
        N+1, & i=N, \\
        i-1, & i\ne N.
    \end{array}  \right.
    \end{array}     
\end{aligned}
\end{equation}

On the other hand, for the second index $\mathbf{v}_{2}=(0,1, \ldots, N-2, N)$, we similarly apply the asymptotic expressions \eqref{smasy4}-\eqref{smasy5} to calculate
\begin{equation}\label{hmspasy3}
\begin{aligned}
    & \det_{1\leq i,j \leq N}\left( \hat{\mathbf{M}}^{[s,+]}(\mathbf{v}_{2}) \right)
    = c_{\mathcal{N}}^{-1} C_{1}^{N(N-1)/2+1}A^{\Gamma-1} W_{\mathcal{N}}'(\bar{z}_{0}) \left[ 1+\cO(A^{-1})\right],\\
    & \det_{1\leq i,j \leq N}\left( \hat{\mathbf{M}}^{[s,-]}(\mathbf{v}_{2}) \right)
    = c_{\mathcal{N}}^{-1} C_{1}^{N(N-1)/2+1}A^{\Gamma-1} \left(W_{\mathcal{N}}'(\bar{z}_{0})\right)^* \left[ 1+\cO(A^{-1})\right].
\end{aligned}
\end{equation}

Therefore, when $A\gg 1$ and $\sqrt{\hat{x}^{2}+\hat{t}^{2}} = \cO(1)$ (i.e., $\sqrt{(x-\bar{x}_{0})^{2}+(t-\bar{t}_{0})^{2}} = \cO(1)$), combining the asymptotic expressions \eqref{hmspasy2} and \eqref{hmspasy3}, we can obtain the dominant term of $A$ in $\det(\mathbf{M}^{[s]})$ near the location $ (\bar{x}_{0}, \bar{t}_{0}) $, as follows:
\begin{equation}\label{homsasy1}
\begin{aligned}
    \det(\mathbf{M}^{[s]}) =  c_{\mathcal{N}}^{-2} C_{1}^{N(N-1)}A^{2(\Gamma-1)} \left|W_{\mathcal{N}}'(\bar{z}_{0}) \right|^2 \left[ \left( u_{1}^{[s,+]}(\hat{x},\hat{t}) + \Delta(\bar{z}_{0}) \right) \left( u_{1}^{[s,-]}(\hat{x},\hat{t}) +\Delta^*(\bar{z}_{0})  \right) +C_{1} \right] \left[ 1+\cO(A^{-1})\right],
\end{aligned}
\end{equation}
where $0\leq s\leq n$, and $\Delta(\bar{z}_{0})=\cO(1)$ is a complex constant defined by
\begin{equation}\label{Delz0}
    \Delta(\bar{z}_{0})=(\alpha_{2}\hat{x}_{0} +\beta_{2}\ii \hat{t}_{0}) \frac{W_{\mathcal{N}}^{[2]}(\bar{z}_{0})}{W_{\mathcal{N}}'(\bar{z}_{0})}.
\end{equation}

Then, as $(\hat{x}, \hat{t})=(x-\bar{x}_{0},t-\bar{t}_{0}) $, we absorb the constant term $\Delta(\bar{z}_{0})$ into $(\bar{x}_{0}, \bar{t}_{0})$, as shown in Eq. \eqref{xt01}, and then generate 
\begin{equation}\label{homsasy2}
\begin{aligned}
    \det(\mathbf{M}^{[s]}) =&c_{\mathcal{N}}^{-2} C_{1}^{N(N-1)}A^{2(\Gamma-1)} \left|W_{\mathcal{N}}'(\bar{z}_{0}) \right|^2 \left[ \left( u_{1}^{[s,+]}(x-\bar{x}_{0},t-\bar{t}_{0})\right) \left( u_{1}^{[s,-]}(x-\bar{x}_{0},t-\bar{t}_{0}) \right) +C_{1} \right]\\
    & \times \left[ 1+\cO(A^{-1})\right], \quad 0\leq s\leq n,
\end{aligned}
\end{equation}
where  
\begin{equation}\label{uspm1}
\begin{aligned}
    &u_{1}^{[0,+]}(x-\bar{x}_{0},t-\bar{t}_{0}) = \ii \chi^{[1]} (x-\bar{x}_{0}+ \chi_{0} (t-\bar{t}_{0})) +h_{1,1}, \quad u_{1}^{[0,-]}(x-\bar{x}_{0},t-\bar{t}_{0})=(u_{1}^{[0,+]}(x-\bar{x}_{0},t-\bar{t}_{0}))^*, \\
    & u_{1}^{[k,+]}(x-\bar{x}_{0},t-\bar{t}_{0}) = \ii \chi^{[1]} (x-\bar{x}_{0}+ \chi_{0} (t-\bar{t}_{0})) +h_{1,1}- {h_{2,1}^{[k]}}, \\ 
    &u_{1}^{[k,-]}(x-\bar{x}_{0},t-\bar{t}_{0})=-\ii (\chi^{[1]})^* (x-\bar{x}_{0}+ \chi_{0}^* (t-\bar{t}_{0})) +h_{1,1}^*+({h_{2,1}^{[k]}})^*, \\
    & h_{1,1}=\frac{\ii \chi^{[1]}}{2\Im(\chi_{0})}, \quad h_{2,1}^{[k]}=\frac{\chi^{[1]}}{\chi_{0}+b_{k}}, \quad 1\leq k\leq n,
\end{aligned}
\end{equation}
and the constants $\chi_{0}$ and $ \chi^{[1]}$ are defined by Eqs. \eqref{chi0} and \eqref{lamchi}, respectively. Obviously, the constant term $\Delta(\bar{z}_{0})$ in $ (\bar{x}_{0}, \bar{t}_{0})$ given in Eq. \eqref{xt01} can be negligible as $A\rightarrow \infty$.

Finally, we prove Eq. \eqref{asym_nmrp} by substituting the above expression \eqref{homsasy2} into the solution formula \eqref{horw} of $q_{k}^{\mathcal{N}}(x,t)$. This completes the proof of Theorem \ref{theo_nmrp}.

\subsection{Proof of Theorem \ref{theo_mrp}}\label{subsec_proofmrp}

This subsection provides the proof of Theorem \ref{theo_mrp}.

First, it is evident that the asymptotic behaviors described in the terms $ (1) $ and $ (3) $ of Theorem \ref{theo_mrp} are consistent with that of NMR-type RW patterns in Theorem \ref{theo_nmrp}, except for the term $ {\Delta}(\bar{z}_{0}) $ in the definition of the first-order RW position $ (\bar{x}_{0}, \bar{t}_{0}) $ given in Eq. \eqref{xt01}. Furthermore, as the large constant $ A\rightarrow \infty $, the effect of the term $ {\Delta}(\bar{z}_{0}) $ on $ (\bar{x}_{0}, \bar{t}_{0}) $ becomes negligible. Consequently, the proofs of the asymptotics for the terms $ (1) $ and $ (3) $ of Theorem \ref{theo_mrp} is omitted here, with {detailed} discussions provided in Sec. \ref{subsec_proofnmrp}.

Next, we present a detailed proof of the asymptotics of MR-type RW patterns in the multiple-root region, as described in the term $ (2) $ of Theorem \ref{theo_mrp}. Based on the definition of the parameters $ \kappa_{l,i} $ in Eq. \eqref{mmrp}, for convenience, we denote $ x_{j}^{[l,\pm]} $ and $ \varphi_{l,j}(z) $ as $ x_{j}^{\pm} $ and $ \varphi_{j}(z) $, respectively, throughout this subsection.

When $ z_{0}-\kappa_{1,1}= \ii \chi^{[1]}(x_{0}+\chi_{0} t_{0})A^{-1} $ is only one $ \Gamma_{0} $-multiple root of $W_{\mathcal{N}}(z)$ with free parameters \eqref{mmrp}, we perform a coordinate transformation:
\begin{equation}\label{hatxt2}
	\tilde{x}=x-\tilde{x}_{0}A, \quad \tilde{t}=t-\tilde{t}_{0}A,
\end{equation}
where $(\tilde{x}_{0}, \tilde{t}_{0})= ({x}_{0}, {t}_{0})A^{-1}$, and $ \chi^{[1]} $, $ \chi_{0} $, and $ \Gamma_{0} $ are defined by Theorem \ref{theo_mrp}. When the internal large parameters $ \{d_{l,j}\}_{1\leq l \leq n, j\in \Lambda_{l}} $ of the solution $ \mathbf{q}^{\mathcal{N}}(x,t) $ are defined by Eq. \eqref{dj2}, we obtain
\begin{equation}\label{xpm3}
\begin{aligned}
	&x_{1}^{+}= y_{1}^{+}+\ii \chi^{[1]} (\tilde{x}_{0}+ \chi_{0} \tilde{t}_{0})A +\kappa_{1,1}A, \quad x_{1}^{-}= y_{1}^{-} -\ii {\chi^{[1]}}^{*} (\tilde{x}_{0}+ \chi_{0}^{*} \tilde{t}_{0}) A +\kappa_{1,1}^{*}A, \\
	&x_{i}^{+}= \begin{cases}
		y_{i}^{+}+ \frac{z_{0}^{i}}{i}A^{i}, &	(n+1) \nmid i,\\
		y_{i}^{+}+ (\alpha_{i} \tilde{x}_{0}+\beta_{i}\ii \tilde{t}_{0})A, & (n+1) \mid i,
	\end{cases} \quad
 	x_{i}^{-}= \begin{cases}
 		y_{i}^{-}+ (\frac{z_{0}^{i}}{i})^{*}A^{i}, &	(n+1) \nmid i,\\
 		y_{i}^{-}+ (\alpha_{i} \tilde{x}_{0}+\beta_{i}\ii \tilde{t}_{0})^{*}A, & (n+1) \mid i, 
 	\end{cases}\\
	&y_{j}^{+}=\alpha_{j} \tilde{x}+\beta_{j}\ii \tilde{t}, \quad y_{j}^{-}= (y_{j}^{+})^{*}, \quad i>1, \quad j\geq 1,
\end{aligned}
\end{equation}
where $ x_{j}^{\pm} $, $ \alpha_{j} $, and $ \beta_{j} $ are given in Eqs. \eqref{xpm1} and \eqref{xpm2}. 
Then, we derive
\begin{equation}\label{smasy6}
	\begin{aligned}
		\sum_{j=0}^{\infty}S_{j}(\mathbf{u}^{[s,l,+]}(v))\varepsilon^{j}=& \exp\left( (\ii \chi^{[1]} (\tilde{x}_{0}+ \chi_{0} \tilde{t}_{0})+\kappa_{1,1}) A\varepsilon+ \sum_{\substack{i=2 \\ (n+1) \nmid i}}^{\infty} \frac{z_{0}^{i}}{i} A^{i}\varepsilon^{i}\right) \exp\left(\tilde{u}_{1}^{[s,+]}(v)\varepsilon +\tilde{u}_{2}^{[s,+]}(v)\varepsilon^2 +\cdots\right) \\
		=&\left(\sum_{i=0}^{\infty} \varphi_{i}({z}_{0}-\kappa_{1,1})A^{i}\varepsilon^{i}  \right) \sum_{j=0}^{\infty}S_{j}(\tilde{\mathbf{u}}^{[s,+]}(v))\varepsilon^{j}, \quad 0\leq s\leq n,
	\end{aligned}
\end{equation}
where $ \mathbf{u}^{[s,l,+]}(v) $ are given in Eq. \eqref{vuspm}, the vectors $ \tilde{\mathbf{u}}^{[s,+]}(v) =(\tilde{u}_{1}^{[s,+]}(v), \tilde{u}_{2}^{[s,+]}(v), \ldots)$ are defined by
\begin{equation}\label{smasy6b}
   \begin{aligned}
	&\tilde{\mathbf{u}}^{[0,+]}(v)=\mathbf{y}^{+} +v\mathbf{h}_{1}, \quad \tilde{\mathbf{u}}^{[k,+]}(v)=\mathbf{y}^{+} +v\mathbf{h}_{1} -\mathbf{h}_{2}^{[k]}, \quad 1\leq k \leq n,\\
	&\mathbf{y}^{+}=(\tilde{y}_{1}^{+}, \tilde{y}_{2}^{+}, \tilde{y}_{3}^{+}, \ldots), \quad 
	\tilde{y}_{j}^{+}=\begin{cases}
		{y}_{j}^{+}, & (n+1) \nmid j,\\
		y_{j}^{+}+ (\alpha_{j} \tilde{x}_{0}+\beta_{j}\ii \tilde{t}_{0})A, & (n+1) \mid j,
	\end{cases}
\end{aligned}	
\end{equation}
with $ \mathbf{h}_{1} $ and $ \mathbf{h}_{2}^{[k]} $ given in Eq. \eqref{xpm1}. It implies that
\begin{equation}\label{sjpm1}
	S_{j}(\mathbf{u}^{[s,l,+]}(v))= \sum_{k=0}^{j} S_{j-k}(\tilde{\mathbf{u}}^{[s,+]}(v))\varphi_{k}(z_{0}-\kappa_{1,1})A^{k}.
\end{equation}

When the integer index vector $ \mathcal{N}= [N_{1}, N_{2}, \cdots, N_{n}] $ satisfies $ N_{1}\geq N_{2}\geq \ldots \geq N_{n} \geq 0  $ and $ \sum_{i=1}^{n} N_{i}=N$, we use the formula \eqref{sjpm1} to expand all elements of the matrices $ \mathbf{M}^{[s,+]} $ $ (1\leq s\leq n) $ in Eq. \eqref{msij1b} for the RW solution $ \mathbf{q}^{\mathcal{N}}(x,t) $. Then, we have
\begin{equation}\label{mspm1}
	\begin{aligned}
		&\mathbf{M}^{[s,+]}= \left( M^{[s,+]}_{1}, M^{[s,+]}_{2}, \ldots, M^{[s,+]}_{n} \right), \\
		& M^{[s,+]}_{l}=\left( C_{1}^{(i-1)/2} \sum_{k=0}^{(n+1)(j-1)+l-i+1} S_{(n+1)(j-1)+l-i+1-k}(\tilde{\mathbf{u}}^{[s,+]}(i-1)) \varphi_{k}(z_{0}-\kappa_{1,1})A^{k} \right)_{ 1\leq i\leq (n+1)N_{1}, 1\leq j\leq N_{l}}
	\end{aligned}
\end{equation}
with $ 1\leq l \leq n $.

Next, we perform a series of column transformations on the matrices $ \mathbf{M}^{[s,+]} $ $ (0\leq s\leq n) $ to eliminate the highest power term of $ A $ in each column. These column transformations are similar to those in Eqs. \eqref{deter2}-\eqref{deter4}, except that the order of the columns is not changed. 

Therefore, when $ N_{1}>N_{2}>\cdots>N_{n}\geq 0 $, we obtain the following simplified matrices by combining with Proposition \ref{prop1} in Appendix \ref{Appe_prop1}, as follows:
\begin{equation}\label{mspm2}
\begin{aligned}
	&\tilde{\mathbf{M}}^{[s,+]}= \left( \tilde{M}^{[s,+]}_{1}, \tilde{M}^{[s,+]}_{2}, \ldots, \tilde{M}^{[s,+]}_{n} \right), \\
	& \tilde{M}^{[s,+]}_{1}=\left( A\varphi_{1}(z_{0}-\kappa_{1,1}) {\mathbf{H}}_{1}^{[s,+]}, A^{2}\varphi_{2}(z_{0}-\kappa_{1,1}){\mathbf{H}}_{2}^{[s,+]}, \ldots, A^{n}\varphi_{n}(z_{0}-\kappa_{1,1}){\mathbf{H}}_{n}^{[s,+]} \right)_{(n+1)N_{1}\times N_{1}},\\ 
	& \tilde{M}^{[s,+]}_{l}=\left( C_{1}^{(i-1)/2} S_{(n+1)(j-1)+l-i+1}(\tilde{\mathbf{u}}^{[s,+]}(i-1)) \right)_{1\leq i\leq (n+1)N_{1}, 1\leq j\leq N_{l}}, \quad 2\leq l \leq n	
\end{aligned}	
\end{equation}
with 
\begin{equation}\label{mspm2b}
\begin{aligned}
		& {\mathbf{H}}_{1}^{[s,+]} = \left( C_{1}^{(i-1)/2}  S_{(n+1)(j-1)-i+1}(\tilde{\mathbf{u}}^{[s,+]}(i-1)) + \cO(A^{-1}) \right)_{1\leq i\leq (n+1)N_{1}, 1\leq j \leq N_{n}+1},\\
		& {\mathbf{H}}_{2}^{[s,+]} = \left( C_{1}^{(i-1)/2}  S_{(n+1)(j+N_{n})-i}(\tilde{\mathbf{u}}^{[s,+]}(i-1)) + \cO(A^{-1}) \right)_{1\leq i\leq (n+1)N_{1}, 1\leq j \leq N_{n-1}-N_{n}},\\
		& {\mathbf{H}}_{3}^{[s,+]} = \left( C_{1}^{(i-1)/2} S_{(n+1)(j+N_{n-1})-i-1}(\tilde{\mathbf{u}}^{[s,+]}(i-1)) + \cO(A^{-1}) \right)_{1\leq i\leq (n+1)N_{1}, 1\leq j \leq N_{n-2}-N_{n-1}},\\
		&\vdots\\
		& {\mathbf{H}}_{n}^{[s,+]} = \left( C_{1}^{(i-1)/2}  S_{(n+1)(j+N_{2})-i-n+2}(\tilde{\mathbf{u}}^{[s,+]}(i-1)) + \cO(A^{n-1}) \right)_{1\leq i\leq (n+1)N_{1}, 1\leq j \leq N_{1}-N_{2}-1}.
\end{aligned}
\end{equation}

Similarly, we simplify the matrices $ \mathbf{M}^{[s,-]} $ $ (1\leq s\leq n) $ in Eq. \eqref{msij1b} to $ \tilde{\mathbf{M}}^{[s,-]} $, expressed as
\begin{equation}\label{mspm3}
	\begin{aligned}
		&\tilde{\mathbf{M}}^{[s,-]}= {\left( ({ \tilde{M}_{1}^{[s,-]} })^{T}, (\tilde{M}^{[s,-]}_{2})^{T}, \ldots, ({\tilde{M}^{[s,-]}_{n}})^{T} \right)}^{T}, \\
		& \tilde{M}^{[s,-]}_{1}=\left( A\varphi_{1}^{*}(z_{0}-\kappa_{1,1}){{\mathbf{H}}_{1}^{[s,-]}}^{T}, A^{2}\varphi_{2}^{*}(z_{0}-\kappa_{1,1}){ {\mathbf{H}}_{2}^{[s,-]}}^{T}, \ldots, A^{n}\varphi_{n}^{*}(z_{0}-\kappa_{1,1}){ {\mathbf{H}}_{n}^{[s,-]}}^{T} \right)_{(n+1)N_{1}\times N_{1}}^{T},\\ 
		& \tilde{M}^{[s,-]}_{l}=\left( C_{1}^{(j-1)/2} S_{(n+1)(i-1)+l-j+1}(\tilde{\mathbf{u}}^{[s,-]}(j-1)) \right)_{1\leq i\leq N_{l}, 1\leq j\leq (n+1)N_{1}}, \quad 2\leq l \leq n
	\end{aligned}	
\end{equation}
with
\begin{equation}\label{mspm3b}
	\begin{aligned}
		& {\mathbf{H}}_{1}^{[s,-]} = \left( {C_{1}}^{(j-1)/2}  S_{(n+1)(i-1)-j+1}(\tilde{\mathbf{u}}^{[s,-]}(j-1)) + \cO(A^{-1}) \right)_{1\leq i\leq N_{n}+1, 1\leq j \leq (n+1)N_{1}},\\
		& {\mathbf{H}}_{2}^{[s,-]} = \left( {C_{1}}^{(j-1)/2}  S_{(n+1)(i+N_{n})-j}(\tilde{\mathbf{u}}^{[s,-]}(j-1)) + \cO(A^{-1}) \right)_{1\leq i\leq N_{n-1}-N_{n}, 1\leq j \leq (n+1)N_{1}},\\
		& {\mathbf{H}}_{3}^{[s,-]} = \left( {C_{1}}^{(j-1)/2} S_{(n+1)(i+N_{n-1})-j-1}(\tilde{\mathbf{u}}^{[s,-]}(j-1)) + \cO(A^{-1}) \right)_{1\leq i\leq N_{n-2}-N_{n-1}, 1\leq j \leq (n+1)N_{1}},\\
		&\vdots\\
		& {\mathbf{H}}_{n}^{[s,-]} = \left( {C_{1}}^{(j-1)/2}  S_{(n+1)(i+N_{2})-j+2-n}(\tilde{\mathbf{u}}^{[s,-]}(j-1)) + \cO(A^{-1}) \right)_{1\leq i\leq N_{1}-N_{2}-1, 1\leq j \leq (n+1)N_{1}},
	\end{aligned}
\end{equation}
and 
\begin{equation}\label{mspm3c}
\begin{aligned}
	&\tilde{\mathbf{u}}^{[0,-]}(v)=\mathbf{y}^{-} +v\mathbf{h}_{1}^*, \quad \tilde{\mathbf{u}}^{[k,-]}(v)=\mathbf{y}^{-} +v\mathbf{h}_{1}^* +({\mathbf{h}_{2}^{[k]}})^* , \quad 1\leq k \leq n,\\
	&\mathbf{y}^{-}=(\tilde{y}_{1}^{-}, \tilde{y}_{2}^{-}, \tilde{y}_{3}^{-}, \ldots), \quad
	\tilde{y}_{j}^{-}=\begin{cases}
		{y}_{j}^{-}, & (n+1) \nmid j,\\
		y_{j}^{-}+ (\alpha_{j}^{*} \tilde{x}_{0}-\beta_{j}^{*}\ii \tilde{t}_{0})A, &  (n+1)\mid j.
	\end{cases}
\end{aligned}
\end{equation}
Thus, we have
\begin{equation}\label{msija2}
	\begin{aligned}
		&\det\left(\mathbf{M}^{[s]}\right)= 
		\begin{vmatrix}
			\mathbf{0}_{N\times N} & -\tilde{\mathbf{M}}^{[s,-]} \\
			\tilde{\mathbf{M}}^{[s,+]} & \mathbb{I}_{(n+1)N_{1}}
		\end{vmatrix},
		 \quad 0 \leq s \leq n.
	\end{aligned}
\end{equation}

Furthermore, we rewrite the determinant \eqref{msija2} of $ \mathbf{M}^{[s]} $ as the following highest-order form of $ A $
\begin{equation}\label{msij3}
	\begin{aligned}
		&\det\left(\mathbf{M}^{[s]}\right)=C_{2} A^{2\Gamma_{1}} 
		\begin{vmatrix}
			\mathbf{0}_{N\times N} & -\bar{\mathbf{M}}^{[s,-]} \\
			\bar{\mathbf{M}}^{[s,+]} & \mathbb{I}_{(n+1)N_{1}}
		\end{vmatrix}
		\left[1+ \cO(A^{-1}) \right] ,  \quad 0 \leq s \leq n,
	\end{aligned}
\end{equation}
where $ C_{2} $ is a constant, $ \Gamma_{1}=(n+1)N_{1}-N-n+1 $, and 
\begin{equation}\label{msij3b}
\begin{aligned}
	&\bar{\mathbf{M}}^{[s,+]}= \left( \bar{M}^{[s, +]}_{1}, \bar{M}^{[s, +]}_{2}, \ldots, \bar{M}^{[s, +]}_{n} \right), \quad \bar{\mathbf{M}}^{[s,-]}= \left(( \bar{M}^{[s, -]}_{1})^{T}, (\bar{M}^{[s, -]}_{2})^{T}, \ldots, (\bar{M}^{[s, -]}_{n})^{T} \right)^{T},\\
	&\bar{M}^{[s,+]}_{1}=\left( C_{1}^{(i-1)/2} S_{(n+1)(j-1)-i+1}(\tilde{\mathbf{u}}^{[s,+]}(i-1)) \right)_{1\leq i\leq (n+1)N_{1}, 1\leq j\leq N_{n}+1},\\
	&\bar{M}^{[s,+]}_{2}=\left( C_{1}^{(i-1)/2} S_{(n+1)(j-1)-i+3}(\tilde{\mathbf{u}}^{[s,+]}(i-1)) \right)_{1\leq i\leq (n+1)N_{1}, 1\leq j\leq N_{1}-1},\\
	&\bar{M}^{[s,+]}_{l}=\left( C_{1}^{(i-1)/2} S_{(n+1)(j-1)-i+1+l}(\tilde{\mathbf{u}}^{[s,+]}(i-1)) \right)_{1\leq i\leq (n+1)N_{1}, 1\leq j\leq N_{l-1}},\\
	&\bar{M}^{[s,-]}_{1}=\left( C_{1}^{(j-1)/2} S_{(n+1)(i-1)-j+1}(\tilde{\mathbf{u}}^{[s,-]}(j-1)) \right)_{1\leq i\leq N_{n}+1, 1\leq j\leq (n+1)N_{1}},\\
	&\bar{M}^{[s,-]}_{2}=\left( C_{1}^{(j-1)/2} S_{(n+1)(i-1)-j+3}(\tilde{\mathbf{u}}^{[s,-]}(j-1)) \right)_{1\leq i\leq N_{1}-1, 1\leq j\leq (n+1)N_{1}},\\
	&\bar{M}^{[s,-]}_{l}=\left( C_{1}^{(j-1)/2} S_{(n+1)(i-1)-j+1+l}(\tilde{\mathbf{u}}^{[s,-]}(j-1)) \right)_{1\leq i\leq N_{l-1}, 1\leq j\leq (n+1)N_{1}}, \quad 3\leq l \leq n.
\end{aligned}
\end{equation}

It is observed that in the first column of matrices $ \bar{\mathbf{M}}^{[s,+]} $ \eqref{msij3}, only the element in the first column is nonzero, which is $ 1 $. Similarly, in the first row of matrices $ \bar{\mathbf{M}}^{[s,-]} $ \eqref{msij3}, only the element in the first column is nonzero, also with a value of $ 1 $. Thus, we can expand the determinant on the right-hand side of Eq. \eqref{msij3} sequentially along the first column and the first row, respectively. This leads to
\begin{equation}\label{msij4}
	\begin{aligned}
	&\det\left(\mathbf{M}^{[s]}\right)=C_{3} A^{2\Gamma_{1}} 
		\begin{vmatrix}
			\mathbf{0}_{(N-1)\times (N-1)} & -\check{\mathbf{M}}^{[s,-]} \\
			\check{\mathbf{M}}^{[s,+]} & \mathbb{I}_{(n+1)\hat{N}_{1}}
		\end{vmatrix}
		\left[1+ \cO(A^{-1}) \right] ,  \quad 0 \leq s \leq n,
	\end{aligned}
\end{equation}
where $ C_{3} $ is a constant, and
\begin{equation}\label{msij4b}
\begin{aligned}
	&\check{\mathbf{M}}^{[s,+]}= \left( \check{M}^{[s, +]}_{1}, \check{M}^{[s, +]}_{2}, \ldots, \check{M}^{[s, +]}_{n} \right), \quad \check{\mathbf{M}}^{[s,-]}= \left(( \check{M}^{[s, -]}_{1})^{T}, (\check{M}^{[s, -]}_{2})^{T}, \ldots, (\check{M}^{[s, -]}_{n})^{T} \right)^{T},\\
	&\check{M}^{[s,+]}_{l}=\left( C_{1}^{(i-1)/2} S_{(n+1)(j-1)-i+1+l}(\tilde{\mathbf{u}}^{[s,+]}(i)) \right)_{1\leq i\leq (n+1)\hat{N}_{1}, 1\leq j\leq \hat{N}_{l}},\\
	&\check{M}^{[s,-]}_{l}=\left( C_{1}^{(j-1)/2} S_{(n+1)(i-1)-j+1+l}(\tilde{\mathbf{u}}^{[s,-]}(j)) \right)_{1\leq i\leq \hat{N}_{l}, 1\leq j\leq (n+1)\hat{N}_{1}}, \quad 1\leq l \leq n
\end{aligned}
\end{equation}
with $ \hat{N}_{1}={N}_{1}-1 $ and $ \hat{N}_{s}=N_{s} $ $ (2\leq s\leq n) $.

In general, when $ N_{1}\geq N_{2}\geq \cdots\geq N_{n}\geq 0 $, we can similarly simplify the determinant of $ \mathbf{M}^{[s]} $ into the form of Eq. \eqref{msij4}, where $ [\hat{N}_{1}, \hat{N}_{1}, \ldots, \hat{N}_{n}] $ is defined by Eq. \eqref{gamma0}. 

Moreover, based on the definition \eqref{schur} of the Schur polynomial, we obtain
\begin{equation}\label{sjpm2}
\begin{aligned}
	S_{j}(\tilde{\mathbf{u}}^{[s,+]}(v)) = \sum_{k=0}^{\lfloor\frac{j}{n+1} \rfloor}S_{j-k(n+1)}(\bar{\mathbf{u}}^{[s,+]}(v)) \frac{\eta_{1}^{k}}{k!}, \quad 0\leq s\leq n,
\end{aligned}
\end{equation}
where the symbol $ \lfloor a \rfloor=\max\{b | b\leq a, b\in \mathbb{N}\} $, $ \eta_{1}=(\alpha_{n+1} \tilde{x}_{0}+\beta_{n+1}\ii \tilde{t}_{0})A $, the vectors $ \bar{\mathbf{u}}^{[s,+]}(v)=\tilde{\mathbf{u}}^{[s,+]}(v)-\eta_{1}\mathbf{e}_{n+1} $, and $ \mathbf{e}_{n+1} $ denotes the infinite-dimensional unit vector with $ 1 $ in the $ (n+1) $-th component.

Then, we employ the formula \eqref{sjpm2} to expand all elements of $ \check{{M}}^{[s,+]}_{l} $ $ (1\leq l\leq n, 0\leq s\leq n) $, and further apply some column transformations to eliminate the highest power term of $ \eta_{1} $ in each column. Thus, we reduce $ \check{M}^{[s,+]}_{l} $ $ (1\leq l\leq n, 0\leq s\leq n) $ to
\begin{equation}\label{mspl1}
	\left( C_{1}^{(i-1)/2} S_{(n+1)(j-1)-i+1+l}(\bar{\mathbf{u}}^{[s,+]}(i-1) +\mathbf{h}_{1}) \right)_{1\leq i\leq (n+1)\hat{N}_{1}, 1\leq j\leq \hat{N}_{l}}.
\end{equation}
Similarly, we also simplify $ \check{M}^{[s,-]}_{l} $ $ (1\leq l\leq n, 0\leq s\leq n) $ to
\begin{equation}\label{mspl2}
	\left( C_{1}^{(j-1)/2} S_{(n+1)(i-1)-j+1+l}(\bar{\mathbf{u}}^{[s,-]}(j-1)+\mathbf{h}_{1}^{*}) \right)_{1\leq i\leq \hat{N}_{l}, 1\leq j\leq (n+1)\hat{N}_{1}}
\end{equation}
with the vectors $ \bar{\mathbf{u}}^{[s,-]}(v)=\tilde{\mathbf{u}}^{[s,-]}(v)-\eta_{1}^{*}\mathbf{e}_{n+1} $.

Therefore, by using the same method, we can simplify all polynomials $ S_{k}(\tilde{\mathbf{u}}^{[s,\pm]}(v)) $ in the matrices $ \check{M}^{[s,\pm]}_{l} $ into $ S_{k}(\check{\mathbf{u}}^{[s,\pm]}(v)) $, where
\begin{equation}\label{cuspm1}
\begin{aligned}
	&\check{\mathbf{u}}^{[s,+]}(v)= \tilde{\mathbf{u}}^{[s,+]}(v)-\sum_{j=1}^{\infty} \eta_{j}\mathbf{e}_{j(n+1)}, \quad 	\check{\mathbf{u}}^{[s,-]}(v)= \tilde{\mathbf{u}}^{[s,-]}(v)-\sum_{j=1}^{\infty} \eta_{j}^{*}\mathbf{e}_{j(n+1)},\\
	&\eta_{j}= (\alpha_{j(n+1)} \tilde{x}_{0}+\beta_{j(n+1)}\ii \tilde{t}_{0})A.
\end{aligned}
\end{equation}

In other words, the vector $ \tilde{\mathbf{u}}^{[s,\pm]}(v) $ in the determinant \eqref{msij4} of $ \mathbf{M}^{[s]} $ can be directly replaced by the vector $ \check{\mathbf{u}}^{[s,+]}(v) +\mathbf{h}_{1}$. This result supports our earlier assertion that the large parameters $ d_{l,j(n+1)} $ $ (j\geq 1) $ have no effect on the RW patterns.

Now, by substituting the determinant \eqref{msij4} of $ \mathbf{M}^{[s]} $ with $ \tilde{\mathbf{u}}^{[s,\pm]}(v) $ replaced by $ \check{\mathbf{u}}^{[s,+]}(v) +\mathbf{h}_{1}$ into the solution formula \eqref{horw}, we can prove that the $ {\mathcal{N}} $-order RW solution $ \mathbf{q}^{\mathcal{N}}(x,t) $ in Theorem \ref{theo_mrp} can be asymptotically reduced to a $ \hat{\mathcal{N}} $-order RW $ \hat{q}_{k}^{\hat{\mathcal{N}}}(x-{x}_{0}, t-{t}_{0}) q^{[0]}_{k}(x,t) $ in the multiple-root region, and its approximation error is $ \cO(A^{-1}) $. Here, $ \hat{q}_{k}^{\hat{\mathcal{N}}}(x-{x}_{0},t-{t}_{0})={q}_{k}^{\hat{\mathcal{N}}}(x-{x}_{0},t-{t}_{0}) \left( q_{k}^{[0]}(x-{x}_{0},t-{t}_{0})\right)^{-1} $ with these internal parameters $ d_{l,j}=h_{1,j} $ $ (j\geq 1) $. This completes the proof of Theorem \ref{theo_mrp}.

\section{Conclusions and Discussions}\label{Sec_Con}
	
In this paper, we identify numerous novel RW patterns for the $ n $-NLSE that incorporate multiple internal large parameters. These patterns are linked to the root structures of specific polynomials. To facilitate a thorough asymptotic analysis of the newly identified RW patterns, we introduce a novel class of polynomials, namely the generalized mixed Adler--Moser (GMAM) polynomials. These polynomials contain multiple arbitrary free parameters, resulting in diverse root structures. The root structures of the polynomials are classified into two main categories: those with only simple roots (non-multiple root, or NMR structures) and those with multiple root (MR structures). Consequently, the RW patterns associated with these GMAM polynomial root structures are categorized into two regions: the simple-root region and the multiple-root region. 

We conduct a detailed asymptotic analysis of high-order RW patterns of the $ n $-NLSE that consist solely of simple-root region, referred to as NMR-type RW patterns. These patterns correspond to the NMR structure of the GMAM polynomials and asymptotically approach a first-order RW at the position corresponding to each simple root of the polynomials. Furthermore, owing to the highly intricate and diverse nature of the MR structures of the GMAM polynomials, this paper focuses on a special case of the MR structure, which contains a single nonzero multiple root along with several simple roots. The asymptotic behavior of the corresponding MR-type RW patterns of the $ n $-NLSE shows that these patterns exhibit several scattered first-order RWs in the simple-root region. In the multiple-root region, they asymptotically approach a lower-order RW. The positions of the first-order and lower-order RWs within the patterns correspond to the simple roots and multiple root of the GMAM polynomials, respectively. In particular, the position of the multiple-root region in the $ (x,t) $-plane for the MR-type RW pattern can be adjusted freely by selecting appropriate values for the multiple root of the GMAM polynomial.

To further illustrate our findings, we present several examples of RW patterns for the $ 3 $-NLSE and $ 4 $-NLSE. These include NMR-type RW patterns exhibiting diverse shapes such as $ 180 $-degree sector, jellyfish-like, and thumbtack-like shapes, as well as MR-type patterns characterized by right double-arrow and right arrow shapes. Their dynamical evolution plots show that these RW patterns are essentially magnified versions of the corresponding polynomial root structures with slight shifts.

It is worth noting that this paper does not consider the high-order RW solutions of the $ n $-NLSE in cases where the characteristic polynomial has different roots, nor does it address RW patterns of the $ n $-NLSE with a single internal large parameter. Additionally, a comprehensive theoretical investigation of all possible MR-type RW patterns remains an open problem. These are intriguing directions for future research and serve as the focus of our subsequent work. Finally, the methodology presented in this paper extends to the study of RW patterns in other multi-component integrable systems, thereby contributing to the broader understanding of such systems and enriching the ongoing research on RW patterns in multi-component integrable models.

\section*{Conflict of interests}
The authors have no conflicts to disclose.

\section*{DATA AVAILABILITY}
Data sharing is not applicable to this article as no new data were created or analyzed in this study.

\section*{Acknowledgments}
Liming Ling is supported by the National Natural Science Foundation of China (No. 12471236) and the Guangzhou Municipal Science and Technology Project (Guangzhou Science and Technology Plan) (No. 2024A04J6245).

\appendix
\renewcommand\thesection{A}
\renewcommand\theequation{\thesection\arabic{equation}} 

\section*{Appendix A}
\setcounter{equation}{0} 

\subsection{The proof of Theorem \ref{Theo1}}\label{Appe_proofRWs}

The high-order vector RW solutions of $n$-NLSE \eqref{nNLSE} were generated from the plane wave seed solution \eqref{seed} by the generalized DT method in Refs. \cite{zhangg2021, linh2024b}. To facilitate the analysis of the asymptotics of the RW patterns, we simplify the expression of the RW solution into the determinant form presented in Theorem \ref{Theo1}. 

From Refs. \cite{zhangg2021, linh2024b}, we can obtain $N$-order solution formula of $n$-NLSE \eqref{nNLSE}, as follows:
\begin{equation}\label{qnfor1}
	q^{[N]}_{k}(x,t)=q^{[0]}_{k}(x,t) + 2\mathbf{Y}_{N,k+1}\bar{\mathbf{M}}^{-1}\mathbf{Y}_{N,1}^{\dagger}, \quad 1\leq k\leq n,
\end{equation}
where
\begin{equation}\label{bmij}
	\begin{aligned}
		&\bar{\mathbf{M}}=\left( \frac{\phi_{i}^{\dagger}\phi_{j}}{\lambda_{i}^{*}-\lambda_{j}}\right)_{1\leq i,j \leq N}, \quad 
        \mathbf{Y}_{N,s}=(\phi_{1,s}, \phi_{2,s}, \ldots, \phi_{N,s}), \quad 1 \leq s \leq n+1, \\	
		&\phi_{j}=(\phi_{j,1}, \phi_{j,2}, \ldots, \phi_{j,n+1})^{T}=\Phi(\lambda_{j};x,t)(c_{j,1}, c_{j,2}, \cdots, c_{j,n+1})^{T},\\
		&\Phi(\lambda;x,t)=\mathrm{diag}\left( 1,\mathrm{e}^{\ii\theta_{1}},\mathrm{e}^{\ii\theta_{2}},\dots,\mathrm{e}^{\ii\theta_{n}}\right)\,\mathbf{E}\,\mathrm{diag}\left( \ee^{{\ii}\omega_{1}}, \ee^{{\ii}\omega_{2}}, \ldots, \ee^{{\ii}\omega_{n+1}} \right),\\
		&\mathbf{E}=
		\begin{pmatrix}
			1 & 1 & \cdots & 1\\
			\frac{ a_{1}}{\chi_{1}+b_{1}} & \frac{ a_{1}}{\chi_{2}+b_{1}} & \cdots & \frac{ a_{1}}{\chi_{n+1}+b_{1}}\\
			\frac{a_{2}}{\chi_{1}+b_{2}} & \frac{ a_{2}}{\chi_{2}+b_{2}} & \cdots & \frac{a_{2}}{\chi_{n+1}+b_{2}}\\
			\vdots &\vdots & \ddots & \vdots\\
			\frac{ a_{n}}{\chi_{1}+b_{n}} & \frac{ a_{n}}{\chi_{2}+b_{n}} & \cdots & \frac{ a_{n}}{\chi_{n+1}+b_{n}}
		\end{pmatrix}, \quad
		\omega_{s}=(\chi_{s}-\lambda)x + (\frac{\chi_{s}^{2}}{2} -\lambda^{2} -\|\mathbf{a}\|_{2}^{2})t +d_{s}(\lambda),
	\end{aligned}
\end{equation}
all $ c_{j,s}$ are arbitrary constants, $ \theta_{k}$ are given in Eq. \eqref{seed}, $\chi_{s}$ are the roots of the characteristic equation \eqref{cheq}, and $d_{s}(\lambda)$ are the arbitrary parameters independent of $x$ and $t$. 

Then, we can simplify the formula \eqref{qnfor1} to the following form
\begin{equation}\label{qkfor2}
	\begin{aligned}
		&q_{k}^{[N]}(x,t)=q^{[0]}_{k}(x,t) \frac{\det(\bar{\mathbf{M}}^{[k]})}{\det(\bar{\mathbf{M}}^{[0]})}, \quad 1\leq k\leq n, \\
	\end{aligned}
\end{equation}
where
\begin{equation}\label{bmsij}
	\begin{aligned}
		&\bar{\mathbf{M}}^{[s]} = \left( \bar{{M}}^{[s]}_{i,j} \right)_{1\leq i,j\leq N}, \quad 0\leq s\leq n, \\
		&\bar{{M}}_{i,j}^{[0]}= \sum_{p,r=1}^{n+1} \frac{c_{i,p}^{*}c_{j,r}}{\chi_{i,p}^{*}-\chi_{j,r}} \ee^{-\ii\omega_{i,p}^{*}+\ii\omega_{j,r}},\\
		&\bar{M}_{i,j}^{[k]}= \sum_{p,r=1}^{n+1} \frac{c_{i,p}^{*}c_{j,r}}{\chi_{i,p}^{*}-\chi_{j,r}} \frac{\chi_{i,p}^{*}+b_{k}}{\chi_{j,r}+b_{k}} \ee^{-\ii\omega_{i,p}^{*}+\ii\omega_{j,r}}.
	\end{aligned}
\end{equation}
	
Next, we introduce a perturbation of the spectral parameter $\lambda_{0}$, as shown in Eq. \eqref{lamchi}, and then consider all eigenvalues $ \lambda \rightarrow \lambda(\varepsilon)$. Furthermore, we perform a Taylor expansion around $ \varepsilon $ to derive a formula of RW solutions for $n$-NLSE \eqref{nNLSE}. Then, through a similar simplification process as in the proof of Theorem $ 3 $ in Ref. \cite{linh2024a}, the above $ \mathcal{N} $-order solution can be simplified into the vector RW solution in Theorem \ref{Theo1}.

\subsection{The first few GMAM polynomials}\label{Appe_tGMAMp}

\begin{equation}\label{ffwhp}
	\begin{aligned}
	&W_{[1,0]}(z)= z+\kappa_{1,1}, \quad W_{[0,1]}(z)=z^{2}+2\kappa_{2,1}z+\kappa_{2,1}^{2} +2\kappa_{2,2},    \quad
	W_{[1,1]}(z)=z^{2}+2 \kappa_{1,1} z +2 \kappa_{1,1} \kappa_{2,1}-\kappa_{2,1}^{2}-2 \kappa_{2,2}, \\ &W_{[2,0]}(z)={z}^{4}+4\,\kappa_{{1,1}}{z}^{3}+ \left( 6\,{\kappa_{{1,1}}}^{2}+
	4\,\kappa_{{1,2}} \right) {z}^{2}+\left( 4\,{\kappa_{{1,1}}}^{
		3}+8\kappa_{{1,1}}\kappa_{{1,2}} \right) z+{\kappa_{{1,1}}}^{4}+4\,
	\kappa_{{1,2}}{\kappa_{{1,1}}}^{2}-4\,{\kappa_{{1,2}}}^{2}-8\,\kappa_{{1,4}}, \\ 
	&W_{[0,2]}(z)= {z}^{6}+6\,\kappa_{{2,1}}{z}^{5}+ 5( 2\kappa_{{2,2}}+3
	\,{\kappa_{{2,1}}}^{2} ) {z}^{4}+ 20\, ( 2\,\kappa_{{2,1}}
	\kappa_{{2,2}}+\,{\kappa_{{2,1}}}^{3} ) {z}^{3} +5\, ( 12\,\kappa_{{2,2}}{\kappa_{{2,1}}}^{2}-8\,\kappa_{{2,4}}+4\,{
		\kappa_{{2,2}}}^{2}\\
	& +3\,{\kappa_{{2,1}}}^{4} ) {z}^{2} +40\,
	( -2\kappa_{{2,1}}\kappa_{{2,4}}+\,{\kappa_{{2,2}}}^{2}\kappa_
	{{2,1}}+\,\kappa_{{2,2}}{\kappa_{{2,1}}}^{3}-2\kappa_{{2,5}}+{\frac{3}{20}}\,{\kappa_{{2,1}}}^{5} ) z+40\,{\kappa_{{2,2}}}^{3}+{
		\kappa_{{2,1}}}^{6}-40\,\kappa_{{2,4}}{\kappa_{{2,1}}}^{2}\\
	&+80\,\kappa_
	{{2,4}}\kappa_{{2,2}}+20\,{\kappa_{{2,2}}}^{2}{\kappa_{{2,1}}}^{2}+10
	\,\kappa_{{2,2}}{\kappa_{{2,1}}}^{4}-80\,\kappa_{{2,1}}\kappa_{{2,5}}
	, \\ 
	&W_{[2,1]}(z)= z^{4}+4 \kappa_{1,1} z^{3}+8 (\frac{1}{2} \kappa_{1,1}^{2}+\frac{1}{2} \kappa_{1,1} \kappa_{2,1}-\frac{1}{4} \kappa_{2,1}^{2}-\frac{1}{2} \kappa_{2,2}) z^{2}+8 (\kappa_{1,1}^{2} \kappa_{2,1}-\frac{1}{2} \kappa_{1,1} \kappa_{2,1}^{2} -\kappa_{1,1} \kappa_{2,2}) z\\
	& -\kappa_{1,1}^{4} +4 \kappa_{1,1}^{3} \kappa_{2,1}-2 \kappa_{1,1}^{2} \kappa_{2,1}^{2}-4 \kappa_{1,2} \kappa_{1,1}^{2}-4 \kappa_{1,1}^{2} \kappa_{2,2}+8 \kappa_{1,1} \kappa_{1,2} \kappa_{2,1}-4 \kappa_{1,2} \kappa_{2,1}^{2}+4 \kappa_{1,2}^{2} -8 \kappa_{1,2} \kappa_{2,2}+8 \kappa_{1,4},\\
	&W_{[1,0,0]}(z)= z+\kappa_{1,1}, \quad  W_{[0,1,0]}(z)=z^{2}+2\kappa_{2,1}z+\kappa_{2,1}^{2} +2\kappa_{2,2},\\
	&W_{[0,0,1]}(z)=z^{3}+3 \kappa_{3,1} z^{2}+(3 \kappa_{3,1}^{2}+6 \kappa_{3,2}) z +\kappa_{3,1}^{3}+6 \kappa_{3,1} \kappa_{3,2}+6 \kappa_{3,3},\\
	&W_{[1,1,0]}(z)=z^{2}+2 \kappa_{1,1} z +2 \kappa_{1,1} \kappa_{2,1}-\kappa_{2,1}^{2}-2 \kappa_{2,2}, \\
	&W_{[1,0,1]}(z)=z^{3}+\frac{3}{2} (\kappa_{3,1} +\kappa_{1,1}) z^{2}+3 \kappa_{1,1} \kappa_{3,1} z +(\kappa_{1,1}-\kappa_{3,1}) (\frac{1}{2}\kappa_{3,1}^{2}+3 \kappa_{3,2})+\kappa_{1,1} \kappa_{3,1}^{2}-3 \kappa_{3,3}, \\ 
	&W_{[0,1,1]}(z)=z^{4}+4 \kappa_{2,1} z^{3}+12 (\frac{1}{4} \kappa_{2,1}^{2}+\frac{1}{2} \kappa_{2,1} \kappa_{3,1}-\frac{1}{4} \kappa_{3,1}^{2}+\frac{1}{2} \kappa_{2,2}-\frac{1}{2} \kappa_{3,2}) z^{2} + ( (6\kappa_{2,1}^{2}-2\kappa_{3,1}^{2}+12 \kappa_{2,2} \\ 
	&  -12 \kappa_{3,2}) \kappa_{3,1} -12\kappa_{3,3} ) z -2 \kappa_{2,1} \kappa_{3,1}^{3}+3 (\kappa_{2,1}^{2}+2 \kappa_{2,2}) \kappa_{3,1}^{2}-12 \kappa_{2,1} \kappa_{3,1} \kappa_{3,2} +6 \kappa_{2,1}^{2} \kappa_{3,2}-12 \kappa_{2,1} \kappa_{3,3}  +12 \kappa_{2,2} \kappa_{3,2},\\ 
	&W_{[2,0,0]}(z)=z^{5}+5 \kappa_{1,1} z^{4}+10 (\kappa_{1,1}^{2}+\kappa_{1,2}) z^{3}+5 (2 \kappa_{1,1}^{3} +6\kappa_{1,1} \kappa_{1,2}+3 \kappa_{1,3}) z^{2}+ 5 ( \kappa_{1,1}^{4} +6\kappa_{1,2} \kappa_{1,1}^{2} +6\kappa_{1,1} \kappa_{1,3}) z \\
	& +\kappa_{1,1}^{5}+10 \kappa_{1,2} \kappa_{1,1}^{3}+15 \kappa_{1,3} \kappa_{1,1}^{2}-30 \kappa_{1,3} \kappa_{1,2}-30 \kappa_{1,5},
	 \\
    &W_{[0,2,0]}(z)= z^{7}+7 \kappa_{2,1} z^{6}+3 (7 \kappa_{2,1}^{2}+6\kappa_{2,2}) z^{5}+5 (7\kappa_{2,1}^{3}+18\kappa_{2,1} \kappa_{2,2} +6 \kappa_{2,3}) z^{4} +5 (7 \kappa_{2,1}^{4} +36 \kappa_{2,2} \kappa_{2,1}^{2} \\
    &   +24 \kappa_{2,1} \kappa_{2,3} +12 \kappa_{2,2}^{2}) z^{3} +180 (\frac{7}{60} \kappa_{2,1}^{5}+ \kappa_{2,2} \kappa_{2,1}^{3}+ \kappa_{2,1}^{2} \kappa_{2,3}+ \kappa_{2,2}^{2} \kappa_{2,1}- \kappa_{2,5}) z^{2} + 360 (\frac{7}{360} \kappa_{2,1}^{6} +\frac{1}{4} \kappa_{2,2} \kappa_{2,1}^{4} \\ 
    &   +\frac{1}{3} \kappa_{2,3} \kappa_{2,1}^{3}+\frac{1}{2} \kappa_{2,2}^{2} \kappa_{2,1}^{2}-\kappa_{2,1} \kappa_{2,5}+\frac{1}{3} \kappa_{2,2}^{3}-\frac{1}{2} \kappa_{2,3}^{2}-\kappa_{2,6}) z +\kappa_{2,1}^{7}+18 \kappa_{2,1}^{5} \kappa_{2,2}+30 \kappa_{2,1}^{4} \kappa_{2,3}+60 \kappa_{2,1}^{3} \kappa_{2,2}^{2} \\
    & -180 \kappa_{2,1}^{2} \kappa_{2,5}+ 60(2 \kappa_{2,2}^{3}-3 \kappa_{2,3}^{2}-6 \kappa_{2,6}) \kappa_{2,1}+360 \kappa_{2,2} (\kappa_{2,2} \kappa_{2,3}+\kappa_{2,5}),
     \\
    &W_{[0,0,2]}(z)= z^{9}+9 \kappa_{3,1} z^{8} +6(6\kappa_{3,1}^{2} +5\kappa_{3,2} )z^{7} +21(4 \kappa_{3,1}^{3}+10 \kappa_{3,1} \kappa_{3,2}+3 \kappa_{3,3} )z^{6} +126( {\kappa_{{3,1}}}^{4}+5{\kappa_{{3,1}}}^{2}\kappa_{{3,2}}    \\
    &  +3\kappa_{ {3,1}}\kappa_{{3,3}} +2{\kappa_{{3,2}}}^{2} ) z^{5} +7( 18{\kappa_{{3,1}}}^{5}+150{\kappa_{{3,1}}}^{3}\kappa_{{3,2}}+135{\kappa_{{3,1}}}^{2}\kappa_{{3,3}} +180\kappa_{{3,1}}{\kappa_{{3,2}}}^{2}+90\kappa_{{3,3}}\kappa_{{3,2}} \\
    &  -90\kappa_{{3,5}} )z^{4}+ 42( 2{\kappa_{{3,1}}}^{6}+25{\kappa_{{3,1}}}^{4}\kappa_{{3,2}}+30{ \kappa_{{3,1}}}^{3}\kappa_{{3,3}}+60{\kappa_{{3,1}}}^{2}{\kappa_{{3,2}}}^{2}+ 60( \kappa_{{3,3}}\kappa_{{3,2}}-\kappa_{{3,5}} ) \kappa_{{3,1}}+20{\kappa_{{3,2}}}^{3} \\
    & -60\kappa_{{3,6}}
    ) z^{3} + ( 36{\kappa_{{3,1}}}^{7}+630{\kappa_{{3,1}}}^{5}\kappa_{{3,2}} +945
    {\kappa_{{3,1}}}^{4}\kappa_{{3,3}}+2520{\kappa_{{3,1}}}^{3}{\kappa_{{3,2}}}^{2}+ 3780( \kappa_{{3,3}}\kappa_{{3,2}}-\kappa_{{3,5}} ) {\kappa_{{3,1}}}^{2}  \\
    & + ( 2520{\kappa_{{3,2}}}^{3}-7560\kappa_{{3,6}} ) \kappa_{{3,1}} +3780\kappa_{{3,3}} {\kappa_{{3,2}}}^{2} -3780\kappa_{{3,7}}) z^{2} +2520 ( {\frac{{\kappa_{{3,1}}}^{8}}{280}}+\frac{1}{12}\kappa_{{3,2}}{\kappa_{{3,1}}}^{6} +{\frac{3\kappa_{{3,3}}{\kappa_{{3,1}}}^{5}}{20}} \\
    & +\frac{1}{2}{\kappa_{{3,2}}}^{2}{\kappa_{{3,1}}}^{4} +{\kappa_{{3,1}}}^{3}\kappa_{{3,3}}\kappa_{{3,2}} -{\kappa_{{3,1}}}^{3}\kappa_{{3,5}}+{\kappa_{{3,1}}}^{2}{\kappa_{{3,2}}}^{3} -3{\kappa_{{3,1}}}^{2}\kappa_{{3,6}} +3\kappa_{{3,1}}\kappa_{{3,3}}{\kappa_{{3,2}}}^{2} -3\kappa_{{3,1}}\kappa_{{3,7}} \\
    & +3{\kappa_{3,3}}^{2}\kappa_{3,2}  +3\kappa_{3,3}\kappa_{3,5} ) z +\kappa_{3,1}^{9} +30{\kappa_{3,1}}^{7}\kappa_{3,2} + 63(\kappa_{3,3}{\kappa_{3,1}}^{6} +4{\kappa_{3,1}}^{5}{\kappa_{3,2}}^{2} +10{\kappa_{{3,1}}}^{4}\kappa_{{3,3}}\kappa_{{3,2}}  -10{\kappa_{3,1}}^{4}\kappa_{3,5}  \\
	& +{\frac {40}{3}}{\kappa_{{3,1}}}^{3}{\kappa_{{3,2}}}^{3} -40{\kappa_{{3,1}}}^{3}\kappa_{{3,6}} +60{\kappa_{{3,1}}}^{2}\kappa_{{3,3}}{\kappa_{{3,2}}}^{2} -60{\kappa_{{3,1}}}^{2}\kappa_{{3,7}} +120\kappa_{{3,1}}{\kappa_{{3,3}}}^{2}\kappa_{{3,2}}  +120\kappa_{{3,1}}\kappa_{{3,3}}\kappa_{{3,5}} -40\kappa_{{3,3}}{\kappa_{{3,2}}}^{3} \\
	&  -120\kappa_{{3,5}}{\kappa_{{3,2}}}^{2} -120\kappa_{{3,2}}\kappa_{{3,7}} +60{\kappa_{{3,3}}}^{3}+120 \kappa_{{3,3}} \kappa_{{3,6}}), \\
    &W_{[1,1,1]}(z)=z^{3}+3 \kappa_{1,1} z^{2}+6 (\kappa_{1,1} \kappa_{2,1}-\frac{1}{2} \kappa_{2,1}^{2}-\kappa_{2,2}) z +\kappa_{3,1}^{3}-3 \kappa_{1,1} \kappa_{3,1}^{2}+6( \kappa_{1,1} \kappa_{2,1}-\frac{1}{2} \kappa_{2,1}^{2}- \kappa_{2,2}+ \kappa_{3,2}) \kappa_{3,1} \\
    & -6 \kappa_{1,1} \kappa_{3,2}+6 \kappa_{3,3}. 
	\end{aligned}
\end{equation}

\subsection{A proposition of the polynomials $ \varphi_{l,j}(z) $}\label{Appe_prop1}

\begin{prop}\label{prop1}
	For the special Schur polynomials $ \varphi_{l,j}(z) $ \eqref{schpj}, if the free parameters $\kappa_{l,i}$ are defined by Eq. \eqref{mmrp}, then the determinants
	\begin{equation}\label{deter1}
		\begin{aligned}
			&\begin{vmatrix}
				\varphi_{l_{1},(n+1)j_{1}+l_{1}}(z_{0}-\kappa_{1,1}) & \varphi_{l_{2},(n+1)j_{2}+l_{2}}(z_{0}-\kappa_{1,1}) \\
				\varphi_{l_{1},(n+1)j_{1}+l_{1}-k}(z_{0}-\kappa_{1,1}) & \varphi_{l_{2},(n+1)j_{2}+l_{2}-k}(z_{0}-\kappa_{1,1})
			\end{vmatrix}, 
		\end{aligned}
	\end{equation}
	all equal to zero with $ 1\leq l_{1},l_{2}\leq n$, $ 1\leq k \leq \min\{l_{1}, l_{2}\}$, and $ j_{1},j_{2} \in \mathbb{N}^{+}$.
\end{prop}

\begin{proof}
	Based on the definition \eqref{schpj} of the Schur polynomials $ \varphi_{l,j}(z) $, when the free parameters $\kappa_{l,i}$ are defined by Eq. \eqref{mmrp}, we have
	\begin{equation}\label{peq1}
		\begin{aligned}
			\sum_{j=0}^{\infty} \varphi_{l,j}(z_{0}-\kappa_{1,1}) \varepsilon^{j} =& \exp\left(\sum_{i=1}^{\infty} \frac{z_{0}^{i}}{i} \varepsilon^{i}\right) \exp\left( -\sum_{i=1}^{\infty}\frac{z_{0}^{i(n+1)}}{i(n+1)} \varepsilon^{i(n+1)} \right) \\
			=&\frac{1}{1-z_{0}\varepsilon} \left(\sum_{i=0}^{\infty} f_{i} z_{0}^{i(n+1)}\varepsilon^{i(n+1)} \right) \\
			=& \left( \sum_{j=0}^{\infty} z_{0}^{j}\varepsilon^{j}\right) \left(\sum_{i=0}^{\infty} f_{i} z_{0}^{i(n+1)}\varepsilon^{i(n+1)} \right) \\
			=& \sum_{j=0}^{\infty}\left( \sum_{i=0}^{\lfloor j/(n+1) \rfloor} f_{i}\right) z_{0}^{j}\varepsilon^{j},
		\end{aligned}
	\end{equation}
	where $ 1\leq l\leq n $, and $ f_{i} $ is defined by
	\begin{equation}\label{peq2}
		\sum_{i=0}^{\infty} f_{i} \varepsilon^{i(n+1)}  = \exp\left( -\sum_{j=1}^{\infty}\frac{1}{j(n+1)} \varepsilon^{j(n+1)} \right).
	\end{equation}
	This derives
	\begin{equation}\label{peq3}
		\varphi_{l,j}(z_{0}-\kappa_{1,1})= \left( \sum_{i=0}^{\lfloor j/(n+1) \rfloor} f_{i}\right) z_{0}^{j},
	\end{equation}
	which subsequently implies that the determinants \eqref{deter1} {are} equal to zero, thereby proving Proposition \ref{prop1}.
\end{proof}

	\bibliographystyle{elsarticle-num}
	\bibliography{Ref_RWofnNLSE}
	
	
	
	
	
	
\end{document}